\documentclass{article}
\usepackage{amsmath, amssymb, amsthm}
\usepackage{setspace}
\usepackage[a4paper, margin = 2.54cm]{geometry}
\usepackage{authblk}
\usepackage{graphicx}
\usepackage{fancyhdr}
\usepackage{color}
\newtheorem{theorem}{Theorem}
\newtheorem{lemma}[theorem]{Lemma}
\newtheorem{corollary}[theorem]{Corollary}

\newtheorem*{acknowledgments}{Acknowledgments}
\numberwithin{equation}{section}
\numberwithin{theorem}{section}
\begin{document}

\title{Hardy-Lieb-Thirring Inequalities for\\Fractional Pauli Operators}
\author{Gonzalo A. Bley, S\o ren Fournais}
\affil{Department of Mathematics, Aarhus University, \authorcr Ny Munkegade 118,  8000 Aarhus C, Denmark}
\maketitle

\begin{abstract}
We provide lower bounds for the sum of the negative eigenvalues of the operator $|\sigma\cdot p_A|^{2s} - C_s/|x|^{2s} + V$ in three dimensions, where $s\in (0, 1]$, covering the interesting physical cases $s = 1$ and $s = 1/2$. Here $\sigma$ is the vector of Pauli matrices, $p_A = p - A$, with $p = -i\nabla$ the three-dimensional momentum operator and $A$ a given magnetic vector potential, and $C_s$ is the critical Hardy constant, that is, the optimal constant in the Hardy inequality $|p|^{2s} \geq C_s/|x|^{2s}$. If spin is neglected, results of this type are known in the literature as Hardy-Lieb-Thirring inequalities, which bound the sum of negative eigenvalues from below by $-M_s\int V_{-}^{1 + 3/(2s)}$, for a positive constant $M_s$. The inclusion of magnetic fields in this case follows from the non-magnetic case by diamagnetism. The addition of spin, however, offers extra challenges that make the result more elusive. It is the purpose of this article to resolve this problem by providing simple bounds for the sum of the negative eigenvalues of the operator in question. In particular, for $1/2 \leq s \leq 1$ we are able to express the bound purely in terms of the magnetic field energy $\|B\|_2^2$ and integrals of powers of the negative part of $V$.
\end{abstract}

\begin{section}{Introduction}
\begin{subsection}{Brief Review of Hardy-Lieb-Thirring Inequalities}
Interest in estimating the sum of the negative eigenvalues of Schr\"{o}dinger operators $p^2 + V$ acting on $L^2(\mathbb{R}^3)$, where $p = -i\nabla$ is the momentum operator and $V$ is a given real-valued potential, has seen an enormous growth since the seminal paper by Lieb and Thirring on the stability of fermionic matter \cite{LT}, where in particular the following inequality was proven
\begin{gather}
\text{Tr}\, (p^2 + V)_{-} \leq C\int_{\mathbb{R}^3}V_{-}^{5/2}\,dx,
\label{equation.lieb.thirring}
\end{gather}
for a certain constant $C$. Here $e_{-} = \max(0, -e)$ is the negative part of a given real number $e$. This is now known as the Lieb-Thirring inequality. The estimate is appealing for several reasons; one is that the right-hand side is simple and easy to compute. Another one, perhaps more important, is that this right side has, in some sense, the correct form: indeed, the most obvious semiclassical estimate for the left side is the integral of the negative part of $p^2 + V$ over all of phase-space, divided by the size of its elementary cell, equal to $h^3 = (2\pi)^3$ in units where $\hbar = 1$,
\begin{gather}
(2\pi)^{-3}\int_{\mathbb{R}^3}\int_{\mathbb{R}^3}\left[p^2 + V(x)\right]_{-}\,dp\,dx = (15\pi^2)^{-1}\int_{\mathbb{R}^3}V_{-}^{5/2}\,dx.
\label{equation.semiclassical}
\end{gather}
Whether $C$ in \eqref{equation.lieb.thirring} can be taken equal to the semiclassical value is actually still an open question. Lieb and Thirring obtained the value $C = 4/(15\pi)$, which follows from an inequality of Schwinger on the number of eigenvalues above a certain a level \cite{S1}.

By using what amounts to the same ideas as in \cite{LT} one can obtain a bound of the form
\begin{gather}
\text{Tr}\,(p^2 + V)_{-}^{\gamma} \leq C\int_{\mathbb{R}^d}V_{-}^{\gamma + d/2}\,dx,
\label{equation.lieb.thirring.non.relativistic}
\end{gather}
where $d$ is now spatial dimension, and $\gamma \geq 1/2$ if $d = 1$, $\gamma > 0$ if $d = 2$, and $\gamma \geq 0$ if $d = 3$. (See \cite{LT2}, a sequel to \cite{LT}.) For $\gamma \geq 3/2$ it is actually known that the optimal constant $C$ is the semiclassical value, meaning here the constant $L_{\gamma, d}^c$ that appears in the computation of the integral
\begin{gather}
(2\pi)^{-d}\int_{\mathbb{R}^d}\int_{\mathbb{R}^d}\left[p^2 + V(x)\right]_{-}^{\gamma}\,dp\,dx = L_{\gamma, d}^c\int_{\mathbb{R}^d}V_{-}(x)^{\gamma + d/2}\,dx,
\end{gather}
equal to
\begin{gather}
2^{-d}\pi^{-d/2}\frac{\Gamma(\gamma + 1)}{\Gamma(\gamma + 1 + d/2)}.
\end{gather}
This was proven for dimension 1 by Lieb and Thirring in \cite{LT2} and for higher dimensions by Laptev and Weidl \cite{LW}; see also the article \cite{BL} by Benguria and Loss  for a different proof of the result of Laptev and Weidl. Another sharp value of the constant $C$ that is known is the case $\gamma = 1/2$ in dimension 1 \cite{HLT}; here the optimal value is $1/2$. For more information regarding inequalities of the type \eqref{equation.lieb.thirring.non.relativistic} we refer the reader to a short encyclopedia article by Lieb \cite{L}. See also, for example, \cite{HLW} and \cite{DLL} for improved constants in certain regimes of $\gamma$ ($1/2 < \gamma < 3/2$ and $1 \leq \gamma < 3/2$, respectively). We remark that finding the optimal constant $C$ in the regime $\gamma < 3/2$ for a general dimension $d \geq 1$ remains an open problem (except in the aforementioned case $d = 1$, $\gamma = 1/2$).

The question of the generalization of the bound \eqref{equation.lieb.thirring.non.relativistic} to fractional laplacians $|p|^{2s}$, $0 < s \leq 1$, was resolved by Daubechies \cite{D}; in particular she obtained the estimate
\begin{gather}
\text{Tr}\,(|p|^{2s} + V)_{-}^{\gamma} \leq L_{s, d, \gamma}\int_{\mathbb{R}^d}V_{-}^{\gamma + d/(2s)}\,dx,
\label{estimate.daubechies}
\end{gather}
for a certain constant $L$. Values of $\gamma$ and $s$ for which \eqref{estimate.daubechies} is valid are as follows: if $d = 1$, $\gamma \geq 1 - 1/(2s)$ for $s > 1/2$, $\gamma > 0$ if $s = 1/2$, and $\gamma \geq 0$ if $s < 1/2$; if $d = 2$, $\gamma \geq 0$ if $s < 1/2$, and $\gamma > 0$ if $s \geq 1/2$; if $d \geq 3$, $\gamma \geq 0$, irrespective of $s$. (For the first condition mentioned, see \cite{F2}.) Her proof is based on a path-integration technique laid out by Lieb in \cite{L1, L2}, through which he derived a bound on the number of eigenvalues of the Schr\"{o}dinger operator $p^2 + V$, acting on $L^2(\mathbb{R}^{d})$, less than or equal to $\alpha \leq 0$, $N_{\alpha}(V)$, namely
\begin{gather}
N_{\alpha}(V) \leq L_d\int_{\mathbb{R}^d}(V(x) - \alpha)_{-}^{d/2}\,dx
\label{equation.CLR}
\end{gather}
for some constant $L_d$. An inequality of the form \eqref{equation.CLR} was also independently discovered by Cwikel \cite{C} and Rosenblum \cite{R}; for this reason \eqref{equation.CLR} is known as the CLR bound.

Next come the Hardy-Lieb-Thirring inequalities. As noticed by Ekholm and Frank \cite{EF}, the Lieb-Thirring inequality can in some situations be a poor estimate on the sum of moments of eigenvalues of a Schr\"{o}dinger operator; they provided as an example the potential $-1/(4|x|^2)$ in 3D combined with the Laplace operator $-\Delta$: from the Hardy inequality $p^2 \geq 1/(4|x|^2)$, it follows that $p^2 + V$ has no bound states, and yet the right side of \eqref{equation.lieb.thirring.non.relativistic} is infinite for all $\gamma \geq 0$. In \cite{EF} the authors provide a remedy to this problem by giving a bound on the sum of moments of the negative eigenvalues of the operator $p^2 - (d - 2)^2/(4|x|^2) + V$ in $d \geq 3$ dimensions of the exact same form as the one without the Hardy term $-(d - 2)^2/(4|x|^2)$, namely
\begin{gather}
\text{Tr}\,\left(-\Delta - \frac{(d - 2)^2}{4|x|^2} + V(x)\right)_{-}^{\gamma} \leq C_{\gamma, d}\int_{\mathbb{R}^d}V(x)_{-}^{\gamma + d/2}\,dx
\label{inequality.ekholm.frank}
\end{gather}
for $\gamma > 0$. We remark here that $(d - 2)^2/4$ is the critical weight $C$ in the non-relativistic Hardy inequality $p^2 \geq C/|x|^2$, in the sense that $p^2 \geq (d - 2)^2/(4|x|^2)$, and if $C > (d - 2)^2/4$, then $p^2 - C/|x|^2$ is unbounded from below. It should be remarked as well that an inequality like \eqref{inequality.ekholm.frank} can only be true for $\gamma > 0$; one reason is that the operator $p^2 - (d - 2)^2/(4|x|^2) + V$ always has a bound state if $V \leq 0$ and $V \neq 0$ \cite{Wei}, while the integral on the right side of \eqref{inequality.ekholm.frank} can be made arbitrarily small. This is why the paper by Ekholm and Frank refers to a so-called ``virtual level.'' (This is in contrast to the corresponding Hamiltonian without the Hardy term, because of, for instance, the CLR bound \eqref{equation.CLR} with $\alpha = 0$ and a $V$ such that the integral on the right side of \eqref{equation.CLR} is small enough.) Inequality \eqref{inequality.ekholm.frank} was then generalized by Frank, Lieb and Seiringer in \cite{FLS} by allowing an arbitrary magnetic field and also arbitrary powers of $|p|$ subject to the restriction $0 < s < \min(1, d/2)$, where now $d \geq 1$. One of their main results is
\begin{gather}
\text{Tr}\,\left(|p - A|^{2s} - \frac{\mathcal{C}_{s, d}}{|x|^{2s}} + V\right)_{-}^{\gamma} \leq L_{\gamma, d, s}\int_{\mathbb{R}^d}V(x)_{-}^{\gamma + d/(2s)}\,dx
\label{equation.HLT.inequality}
\end{gather}
for $\gamma > 0$ and a certain constant $L$ (independent of the magnetic potential $A$ and $V$). Here
\begin{gather}
\mathcal{C}_{s, d} \equiv 2^{2s}\frac{\Gamma\left[\left(d + 2s\right)/4\right]^2}{\Gamma\left[\left(d - 2s\right)/4\right]^2}
\label{equation.constant.hardy.inequality}
\end{gather}
is the critical constant in the fractional Hardy inequality $|p|^{2s} \geq \mathcal{C}_{s, d}/|x|^{2s}$ \cite{H}. Their generalization is non-trivial on two accounts: first, $(-\Delta)^{2s}$ for $0 < s < 1$ is non-local, as opposed to $-\Delta$; and second, the magnetic field may not be added easily into the inequality by using the techniques in \cite{EF}, which is to be contrasted to the situation without the term $1/|x|^{2s}$, where the inclusion of $A$ follows directly from diamagnetism at the level of the semigroup $e^{-tH}$, by first obtaining a bound on the number of bound states below a certain level, using path integration methods, as in \cite{L1, L2}. Another important article in this direction is one by Frank \cite{F}, where he provides a substantially simpler proof of \eqref{equation.HLT.inequality}, as well as extending the range of the power $s$ from $0 < s < \min(1, d/2)$ to $0 < s < d/2$, and providing an easy proof that every Lieb-Thirring inequality implies a magnetic version of it, possibly with a worse constant. We refer the reader to \cite{F2} for a review of fractional Lieb-Thirring and Hardy-Lieb-Thirring inequalities.
\end{subsection}
\begin{subsection}{The Inclusion of Spin and the Pauli Operator}
In the previous paragraph we have  provided a brief overview of Lieb-Thirring and Hardy-Lieb-Thirring inequalities, starting from the case of the 3D Laplacian and eventually getting to the fractional Laplacian with a magnetic field and the Hardy term added. None of these works has considered the inclusion of spin, however. When including spin, one passes from $p_A^2 \equiv (p - A)^2$ to the Pauli operator $\left(\sigma\cdot p_A\right)^2$, with $\sigma\cdot p_A = \sigma_1 p_A^1 + \sigma_2 p_A^2 + \sigma_3 p_A^3$ being the Dirac operator. Here the operators act on spinors $(\psi_1, \psi_2)^T$, and $\sigma \equiv (\sigma_1, \sigma_2, \sigma_3)$ is the vector of Pauli matrices
\begin{gather}
\sigma_1 \equiv \begin{pmatrix}0 & 1\\1 & 0\end{pmatrix}, \qquad \sigma_2 \equiv \begin{pmatrix}0 & -i\\i & 0\end{pmatrix}, \qquad \sigma_3 \equiv \begin{pmatrix}1 & 0\\0 & -1\end{pmatrix}.
\end{gather}

The main object that has been studied in this regard is the sum of the negative eigenvalues of the non-relativistic operator $(\sigma\cdot p_A)^2 + V$. An important estimate on this for a general magnetic field $B$ was obtained by Lieb, Loss and Solovej in a paper whose main objective was to prove stability of matter with the Pauli operator $(\sigma\cdot p_A)^2$ \cite{LLS}; they obtained in particular the bound
\begin{gather}
\text{Tr}\,\left[(\sigma\cdot p_A)^2 + V\right]_{-} \leq A_1\int V_{-}(x)^{5/2}\,dx + A_2\left(\int B(x)^2\,dx\right)^{3/4}\left(\int V_{-}(x)^4\,dx\right)^{1/4},
\label{inequality.LLS}
\end{gather}
for certain constants $A_1, A_2$.

There have been other bounds for the sum of the negative eigenvalues of $(\sigma\cdot p_A)^2 + V$; in particular, those obtained by Erd\"{o}s \cite{E1}, Erd\"{o}s and Solovej \cite{ES2, ES3}, Sobolev \cite{So1, So2}, and Bugliaro, Fefferman, Fr\"{o}hlich, Graf, and Stubbe \cite{BFFGS}. We note in passing that the first Lieb-Thirring-type estimate for $\text{Tr}\,\left[(\sigma\cdot p_A)^2 + V\right]_{-}$ can be found in a series of papers by Lieb, Solovej and Yngvason \cite{LSY1, LSY2, LSY3}, but is valid only for a constant magnetic field; the non-constant case was then covered in the aforementioned paper \cite{LLS}. Regarding the relativistic Pauli operator $|\sigma\cdot p_A|$, we would like to single out a paper by Erd\"{o}s, Fournais, and Solovej \cite{EFS}. There, two Lieb-Thirring-type inequalities were obtained for $\sqrt{\beta(\sigma\cdot p_A)^2 + \beta^2} - \beta$ (where $\beta > 0$, and the units are such that $\hbar$, and the mass and the charge of the electron are all equal to 1). One of them, \cite[Theorem 2.2]{EFS}, was
\begin{align}
\text{Tr}\,\Big[\sqrt{\beta(\sigma\cdot p_A)^2 + \beta^2} &- \beta + V\Big]_{-}\nonumber\\
\leq & \, C\left(\int V_{-}(x)^{5/2}\,dx + \beta^{-3/2}\int V_{-}(x)^4\,dx + \left(\int B^2\right)^{3/4}\left(\int V_{-}^4\right)^{1/4}\right),
\end{align}
which should be contrasted with \eqref{inequality.LLS}: everything remains the same, except for an additional term, the integral of $V_{-}^4$, that one may dub the ``relativistic term,'' in view of Daubechies' estimate \eqref{estimate.daubechies}. As $\beta \to \infty$, the middle term vanishes, which is consistent with the fact that $\sqrt{\beta(\sigma\cdot p_A)^2 + \beta^2} - \beta \to (\sigma\cdot p_A)^2/2$, thus recovering \eqref{inequality.LLS} (up to a global constant). The second estimate, \cite[Theorem 2.3]{EFS}, was
\begin{align}\label{estimate.local.PHLT.inequality}
 &\text{Tr}\,\Big[\phi_r\Big(\sqrt{\beta(\sigma\cdot p_A)^2 + \beta^2} - \beta - \frac{1}{|x|} + V\Big)\phi_r\Big]_{-}\\
&\qquad\qquad \leq  \, C\left(\eta^{-3/2}\int B^2 + \eta^{-3}r^3 + \eta^{-3/2}\int V_{-}^{5/2} + \eta^{-3}\beta^3\int V_{-}^4 + \left(\int B^2\right)^{3/4}\left(\int V_{-}^4\right)^{1/4}\right),\nonumber
\end{align}
where $\phi_r$ is any compactly supported real-valued function satisfying $\|\phi_r\|_{\infty} \leq 1$, $\beta > \pi^2/4$, and $\eta = (1 - \pi^2/(4\beta))/10$. The result is rather restrictive on several accounts, mainly that $\phi$ is not allowed to be unbounded (far less be equal to 1) and $\beta$ must be bigger than $\pi^2/4$. The present work was partly motivated by the limitations of \eqref{estimate.local.PHLT.inequality}; in particular, we have been able to obtain a global result concerning the massless relativistic Pauli operator with a critical Hardy term, which is our first theorem in the next subsection.
\end{subsection}
\begin{subsection}{Main Results}
In this subsection we shall state the central results of the article. The following is the main theorem, a Hardy-Lieb-Thirring inequality for the operator $|\sigma\cdot p_A|^{2s}$, $1/2 \leq s \leq 1$, in terms of the magnetic field energy $\|B\|_2^2$ and integrals of powers of the negative of the potential:
\begin{theorem}[Pauli-Hardy-Lieb-Thirring Inequality]
\label{theorem.main.result}
There are positive constants $P_s$, $Q_s$, and $R_s$, depending exclusively on $1/2 \leq s \leq 1$, such that
\begin{align}
\text{\textnormal{Tr}}\left(|\sigma\cdot p_A|^{2s} - \frac{C_s}{|x|^{2s}} + V\right)_{-} \leq & \, P_s\int_{\mathbb{R}^3}V_{-}^{1 + 3/(2s)}\,dx + Q_s\left(\int_{\mathbb{R}^3}|B|^2\,dx\right)^{2s}\nonumber\\
& \, + R_s\left(\int_{\mathbb{R}^3}|B|^2\,dx\right)^{3/4}\left(\int_{\mathbb{R}^3} V_{-}^4\,dx\right)^{1/4},
\label{equation.PHLT.inequality}
\end{align}
where $C_s \equiv \mathcal{C}_{s, 3}$, defined in Equation \eqref{equation.constant.hardy.inequality}, is the critical constant in the Hardy inequality in 3D.
\end{theorem}
Theorem \ref{theorem.main.result} includes the physical cases $s = 1$ and $s = 1/2$, but not the smaller powers $0 < s < 1/2$. The purpose of the next theorem is to cover those powers, as well as to include the previous ones.
\begin{theorem}[Second Pauli-Hardy-Lieb-Thirring Inequality]
\label{theorem.second.main.result}
There are positive constants $M_s$ and $N_s$, depending exclusively on $0 < s \leq 1$, such that
\begin{gather}
\text{\textnormal{Tr}}\left(|\sigma\cdot p_A|^{2s} - \frac{C_s}{|x|^{2s}} + V\right)_{-} \leq M_s\int_{\mathbb{R}^3}|B|^{s + 3/2}\,dx + N_s\int_{\mathbb{R}^3}V_{-}^{1 + 3/(2s)}\,dx.
\label{equation.PHLT.second.inequality}
\end{gather}
\end{theorem}
Both theorems match in form for $s = 1/2$. As will be clear from the proofs of Theorems \ref{theorem.main.result} and \ref{theorem.second.main.result}, the content of Section \ref{section.proof.main.result}, it is possible to obtain different inequalities than \eqref{equation.PHLT.inequality} or \eqref{equation.PHLT.second.inequality}, with different powers of the magnetic field $B$. When stating our results, however, we have not aimed for maximum generality, but just to provide what we believe are the simplest and most important inequalities one can obtain with the techniques we use. For that reason, we have only stated inequalities \eqref{equation.PHLT.inequality} and \eqref{equation.PHLT.second.inequality} as our main results. We remark here that all the constants appearing in the theorems, namely $P_s$, $Q_s$, $R_s$, $M_s$, and $N_s$, can be written out explicitly by carefully keeping track of how constants change as the proofs progress; however, we do not do this, and we never show the constants in closed form.

Theorems \ref{theorem.main.result} and \ref{theorem.second.main.result} imply in particular bounds on the spectrum of $|\sigma\cdot p_A|^{2s} - C_s/|x|^{2s} + V$, namely
\begin{corollary}[Bounds on the Spectrum of Fractional Pauli Operators]
For $1/2 \leq s \leq 1$,
\begin{align}
|\sigma\cdot p_A|^{2s} - \frac{C_s}{|x|^{2s}} + V \geq & \, -P_s\int_{\mathbb{R}^3}V_{-}^{1 + 3/(2s)}\,dx - Q_s\left(\int_{\mathbb{R}^3}|B|^2\,dx\right)^{2s}\nonumber\\
& \, - R_s\left(\int_{\mathbb{R}^3}|B|^2\,dx\right)^{3/4}\left(\int_{\mathbb{R}^3} V_{-}^4\,dx\right)^{1/4},
\end{align}
and for $0 < s \leq 1$,
\begin{gather}
|\sigma\cdot p_A|^{2s} - \frac{C_s}{|x|^{2s}} + V \geq - M_s\int_{\mathbb{R}^3}|B|^{s + 3/2}\,dx - N_s\int_{\mathbb{R}^3}V_{-}^{1 + 3/(2s)}\,dx.
\end{gather}
\end{corollary}
By setting $V = 0$ and $s = 1/2$, we see that an ultrarelativistic electron (i.e. massless) tied to a central Coulomb force field is stable when the atomic number $Z$ is allowed to go all the way to the critical Hardy constant $C_s$, as long as the field energy is added with a sufficiently large constant $\beta > 0$ in front, in other words
\begin{corollary}[Stability of the Relativistic One-electron Atom with Magnetic Fields and Spin]
\label{corollary.stability.relativistic}
For a sufficiently large $\beta > 0$,
\begin{gather}
|\sigma\cdot p_A| - \frac{C_{1/2}}{|x|} + \beta\int_{\mathbb{R}^3}|B|^2\,dx \geq 0.
\label{inequality.stability.matter}
\end{gather}
\end{corollary}
Corollary \ref{corollary.stability.relativistic} is a non-trivial statement for several reasons: first, the form $|\sigma\cdot p_A| - C_{1/2}/|x|$ is not positive (because, for example, of the existence of zero modes, i.e. pairs $(\psi, A)$ of non-zero $L^2$ functions $\psi$ and magnetic potentials $A$ with $B = \nabla \times A \in L^2$,  such that $\sigma\cdot p_A \psi = 0$ \cite{LY}); second, we are allowing the constant in front of $1/|x|$ to go all the way to the critical Hardy constant; and third, in view of the scaling $\psi(x) \to \lambda^{3/2}\psi(\lambda x)$, $A(x) \to \lambda A(\lambda x)$, the infimum of
\begin{gather}
\left(\psi, \left(|\sigma\cdot p_A| - \frac{C_{1/2}}{|x|} + \beta\int_{\mathbb{R}^3}|B|^2\,dx\right)\psi\right)
\end{gather}
over all magnetic vector potentials $A$ and normalized functions $\psi$ is either $-\infty$ or 0. Result \eqref{inequality.stability.matter} shows in particular that 0 is the correct answer in the property just mentioned. It is the relativistic analog of the result on the stability of Coulomb systems with magnetic fields (where $(\sigma\cdot p_A)^2$ is considered, instead of $|\sigma\cdot p_A|$) in a series of papers by Fr\"{o}hlich, Lieb and Loss \cite{FLL}, Lieb and Loss \cite{LL2}, and Loss and Yau \cite{LY}.
\end{subsection}
\begin{subsection}{The Strategy of Proof and the Main Tools Used}
In this subsection we shall summarize the main elements of the proofs of the results stated above. We shall concentrate here mostly on explaining the ideas behind the proof of Theorem \ref{theorem.second.main.result} for $0 < s < 1$, since they already convey most of the concepts and techniques used. First, the fractional Pauli operator $|\sigma\cdot p_A|^{2s}$, $0 < s < 1$, is replaced by the simpler operator $|p_A|^{2s}$ (which amounts to the removal of the spin variable), and the cost involved in doing so is estimated explicitly. More precisely, for functions $\psi$ normalized in $L^2$, a bound of the form
\begin{gather}
\left|\left(\psi,\left(|\sigma\cdot p_A|^{2s} - |p_A|^{2s}\right)\psi\right)\right| \leq \varepsilon\||p_A|^u\psi\|_2^{2s/u} + \Omega\|B\|_r^{2sr/(2r - 3)}
\label{equation.first.estimate}
\end{gather}
is obtained, where $\varepsilon > 0$ is arbitrarily small, $u$ and $r$ are positive numbers satisfying certain conditions (which will be made precise in due time), and $\Omega$ is a positive function of all the parameters involved; in particular it diverges as $\varepsilon \to 0$. The key is that we will pick $u < s$; the first error term on the right side of \eqref{equation.first.estimate} involves then a lower-order power of $|p_A|$, which we will be able to control, as explained shortly. Equation \eqref{equation.first.estimate} is the content of our Estimate I: Fractional Case, Theorem \ref{theorem.estimate.I} below.

After this first reduction, we obtain a second estimate, where $|p_A|^{2s}$ is replaced by $|p|^{2s}$ (the magnetic field is removed) at a price that can again be estimated explicitly, namely (still for normalized $\psi$)
\begin{gather}
\left|\left(\psi,\left(|p_A|^{2s} - |p|^{2s}\right)\psi\right)\right| \leq \varepsilon\||p|^u\psi\|_2^{2s/u} + \mathcal{J}\|B\|_r^{2sr/(2r - 3)},
\label{equation.second.estimate}
\end{gather}
in entire analogy to the first estimate. Equation \eqref{equation.second.estimate} appears in Theorem \ref{theorem.estimate.II}, below. These two estimates follow essentially from computations involving the resolvent expansion for a fraction of a positive self-adjoint operator $A$, $A^{\alpha} = \left[\sin(\pi \alpha)/\pi\right]\int_0^{\infty}A(A + a)^{-1}a^{\alpha - 1}\,da$, $0 < \alpha < 1$.

The combination of the two estimates then yields, in particular, an estimate of the form
\begin{gather}
\left|\left(\psi,\left(|\sigma\cdot p_A|^{2s} - |p|^{2s}\right)\psi\right)\right| \leq \varepsilon\||p|^u \psi \|_2^{2s/u} + \mathcal{L}\int_{\mathbb{R}^3}|B|^{s + 3/2}\,dx,
\label{equation.third.estimate}
\end{gather}
with $\mathcal{L}$ diverging as $\varepsilon \to 0$. ($r$ was picked so that we ended up with a simple integral of a power of $|B|$.) This last estimate implies, in particular
\begin{gather}
\left(\psi, \left(|\sigma\cdot p_A|^{2s} - \frac{C_s}{|x|^{2s}}\right)\psi\right) \geq 
\left(\psi, \left(|p|^{2s} - \frac{C_s}{|x|^{2s}}\right)\psi\right) - \varepsilon\||p|^u \psi \|_2^{2s/u} - \mathcal{L}\int_{\mathbb{R}^3}|B|^{s + 3/2}\,dx,
\end{gather}
for $\psi$ normalized in $L^2$.

The key now is that the entire form $|p|^{2s} - C_s/|x|^{2s}$ can control $\varepsilon\||p|^u \psi \|_2^{2s/u}$ for any $ 0 < u < s$ and sufficiently small $\varepsilon$; more precisely, for all $\psi$ with $\| \psi \|_2 = 1$,
\begin{gather}
\left(\psi, \left(|p|^{2s} - \frac{C_s}{|x|^{2s}}\right)\psi\right) \geq H_{s, u}\||p|^u \psi \|_2^{2s/u}
\label{inequality.SSS}
\end{gather}
for a certain constant $H_{s, u} > 0$. This was noticed for the first time by Solovej, S\o rensen and Spitzer for the relativistic case $s = 1/2$ \cite[Theorem 2.3]{SSS}; it was later generalized by Frank to any $0 < s < 3/2$ \cite[Theorem 1.2]{F}. The scale-invariant form in which we write the inequality here appears in \cite{F}. After an application of \eqref{inequality.SSS} one finds an estimate for quadratic forms,
\begin{gather}
|\sigma\cdot p_A|^{2s} - \frac{C_s}{|x|^{2s}} \geq \lambda\left(|p|^{2s} - \frac{C_s}{|x|^{2s}}\right) - \mathcal{L}\int_{\mathbb{R}^3}|B|^{s + 3/2}\,dx,
\label{inequality.fundamental}
\end{gather}
with $0 \leq \lambda < 1$ and $\mathcal{L} \to \infty$ as $\lambda \to 1$. Inequality \eqref{inequality.fundamental} is fundamental to our work -- it allows us to replace $|\sigma\cdot p_A|^{2s}$ by $|p|^{2s}$ up to a lower bound, a constant, and an integral of a power of $|B|$. The reader at this point should notice that setting $\lambda = 0$ and $s = 1/2$ proves the stability of relativistic matter claimed in Corollary \ref{corollary.stability.relativistic}. Equation \eqref{inequality.fundamental} appears in Theorem \ref{theorem.quadratic.form}.

From \eqref{inequality.fundamental} the proof of the Second Pauli-Hardy-Lieb-Thirring Inequality, Theorem \ref{theorem.second.main.result}, goes roughly as follows: we first prove a localization estimate for the fractional Pauli operator, namely
\begin{gather}
|\sigma\cdot p_A|^{2s} \geq \sum_{n = 0}^{\infty}\varphi_n\left(|\sigma\cdot p_A|^{2s} - D_n^{2s}\right)\varphi_n,
\label{equation.localization.estimate.introduction}
\end{gather}
where $D_n$ is equal to a constant times $2^{-n}$, and the functions $\varphi_n$ form a dyadic partition of unity. (Where $\sum_{n = 0}^{\infty}\varphi_n^2 = 1$, and the functions $\varphi_n$ have support on shells growing roughly as $2^n$ in width and distance from the origin.) Equation \eqref{equation.localization.estimate.introduction} is found in our Theorem \ref{theorem.localization.estimate}. What makes this estimate possible is a so-called pull-out formula
\begin{gather}
\left(\sum_{n = 0}^{\infty}S_n A_n S_n\right)^s \geq \sum_{n = 0}^{\infty}S_n A_n^s S_n,
\label{equation.pull.out.formula}
\end{gather}
valid for any collection of positive self-adjoint operators $A_n$ and bounded positive self-adjoint operators $S_n$ such that $\sum_{n = 0}^{\infty}S_n^2 = 1$. This estimate appears for the case $s = 1/2$ in \cite[Lemma 3.1]{EFS}. Its proof follows from the integral representation for the $s$-power of a positive self-adjoint operator $A$ (above Equation \eqref{equation.third.estimate}) and an inequality for resolvents similar in spirit to \eqref{equation.pull.out.formula}, which appeared first in \cite{BFFGS}; see also \cite[Proposition 6.1]{ES1}. We remark also that \eqref{equation.pull.out.formula} is an immediate application of a result of F. Hansen and G.K. Pedersen \cite[Theorem 2.1]{HP}, as $x \mapsto -x^s$ is operator convex for $0 < s < 1$.

The localization estimate \eqref{equation.localization.estimate.introduction} is one of the most useful technical tools in the present article: the careful reader will notice that this is an estimate for a nonlocal operator with a local error term. This particular feature turns out to be extremely convenient, and sidesteps important complications that would arise had we used an estimate with a non-local error: the reader here should compare \eqref{equation.localization.estimate.introduction} with estimates where the error is nonlocal: see, for example, an exact formula due to M. Loss in \cite[Theorem 9]{LY2}, \cite[Lemma 3.5]{FLS} and \cite[Theorem 3.1]{FLS2}, or an estimate due to E. Lenzmann and M. Lewin in \cite[Lemma A.1]{LL2}. Estimate \eqref{equation.localization.estimate.introduction} allows us in particular to split the trace of the negative part of $|\sigma\cdot p_A|^{2s} - C_s/|x|^{2s} + V$ as the sum of that of two localized Hamiltonians -- one which is localized around the origin and one far away from it. For the first Hamiltonian we apply Estimate \eqref{inequality.fundamental} and then apply the Hardy-Lieb-Thirring inequality \eqref{equation.HLT.inequality} as appears, for example, in \cite{F}. For the part far away we obtain an effective localized potential, roughly equal to $-E\phi /|x|^{2s} + V$, with $E$ a constant that goes beyond $C_s$ and $\phi$ the indicator function of a region far from the origin. In this way the singularity at the origin of $1/|x|^{2s}$ is removed and we can treat all of $-E\phi /|x|^{2s} + V$ as a potential; we then use the BKS inequality, which in the form that is relevant to us here states that
\begin{gather}
\text{Tr}\,(A - B)_{-} \leq \text{Tr}\,(A^{1/s} - B^{1/s})_{-}^s
\label{inequality.BKS}
\end{gather}
for any positive self-adjoint operators $A$ and $B$ and $0 < s < 1$. This inequality is due to Birman, Koplienko and Solomyak \cite{BKS}, but see, in particular, \cite{LSS, LSS2} for a simple proof of the version stated here. Estimate \eqref{inequality.BKS} allows us then to effectively replace $|\sigma\cdot p_A|^{2s}$ by $(\sigma\cdot p_A)^2 = p_A^2 - \sigma\cdot B$; by treating the term $\sigma\cdot B$ as a potential we can then use the non-relativistic Lieb-Thirring inequality \eqref{equation.lieb.thirring} (for the magnetic momentum $p_A$) to obtain an estimate for the term with the singularity of $1/|x|^{2s}$ removed. This argument forms the core of the proof of the Second Pauli-Hardy-Lieb-Thirring Inequality, Theorem \ref{theorem.second.main.result}.

We briefly explain now how the proof of the Pauli-Hardy-Lieb-Thirring Inequality, Theorem \ref{theorem.main.result}, differs from the second one, Theorem \ref{theorem.second.main.result}. Again, we shall restrict ourselves to $1/2 \leq s < 1$ for the moment. We recall that the second inequality was proven by using Estimate \eqref{equation.localization.estimate.introduction} to obtain two Hamiltonians, one localized around the origin and another one far away from it. For the first one we use an estimate of the form
\begin{gather}
|\sigma\cdot p_A|^{2s} - \frac{C_s}{|x|^{2s}} \geq \lambda\left(|p|^{2s} - \frac{C_s}{|x|^{2s}}\right) - \mathcal{M}\left(\int_{\mathbb{R}^3}|B|^2\,dx\right)^{2s},
\label{inequality.fundamental.power.2}
\end{gather}
which we prove in essentially the same way as Equation \eqref{inequality.fundamental}. For the second Hamiltonian, the one localized away from the origin, we utilize, and prove in Theorem \ref{theorem.pauli.lieb.thirring.power.2} below, a generalization of Equation \eqref{inequality.LLS} for the powers $1/2 \leq s \leq 1$, namely
\begin{gather}
\text{Tr}\,(|\mathcal{P}_A|^{2s} + W)_{-} \leq U_s \int_{\mathbb{R}^3}W_{-}^{1 + 3/(2s)}\,dx + V_s\left(\int_{\mathbb{R}^3}|B|^2\,dx\right)^{3/4}\left(\int_{\mathbb{R}^3}W_{-}^4\,dx\right)^{1/4},
\end{gather}
which completes the argument behind the Pauli-Hardy-Lieb-Thirring Inequality for $1/2 \leq s < 1$.

Finally, the two Pauli-Hardy-Lieb-Thirring inequalities with $s = 1$ are proven the same way as before, only that the estimate
\begin{gather}
|\mathcal{P}_A|^2 - \frac{C_1}{|x|^2} \geq \lambda\left(|p_A|^2 - \frac{C_1}{|x|^2}\right) - D_r\|B\|_r^{2r/(2r - 3)},
\end{gather}
with $3/2 < r \leq \infty$, the content of Theorem \ref{theorem.quadratic.form.s.1} below, is used, in lieu of \eqref{inequality.fundamental} or \eqref{inequality.fundamental.power.2}.

The rest of the article will be devoted to explaining in detail the steps of the proofs and the technicalities involved.
\end{subsection}
\begin{subsection}{The Structure of the Article and Acknowledgments}
The article is structured as follows: Section 2 is devoted exclusively to the estimates for the differences between $|\sigma\cdot p_A|^{2s}$ and $|p_A|^{2s}$ and between $|p_A|^{2s}$ and $|p|^{2s}$, as explained in the previous subsection. We call these Estimate I and II, respectively. In Section 3 we provide technical lemmas that will allow us to prove the Pauli-Hardy-Lieb-Thirring inequalities. In particular, it is here where we provide the proof of the main localization estimate for the fractional Pauli operator $|\sigma\cdot p_A|^{2s}$, as explained above. Some technical lemmas in the computation of the trace of the negative part of certain operators, which will be used later, are also provided here. Section 4 is then dedicated to the proofs of the Pauli-Hardy-Lieb-Thirring inequalities, which proceed smoothly, as many of the tools used have been proven already in previous sections. We provide at the end of the article an appendix, where we give explicit values for the constants appearing in Estimates I and II from Section 2, and a statement and an easy proof of a Sobolev inequality for rotors that we use in the proof of Estimate II.
\begin{acknowledgments}
The authors were partially supported by a Sapere Aude grant from the Independent Research Fund Denmark, Grant number DFF--4181-00221. They would also like to thank the anonymous referee for precise and useful comments that helped make the article better, in particular for pointing out the article by Hansen and Pedersen cited by Equation \eqref{equation.pull.out.formula}, of which the authors were unaware.
\end{acknowledgments}
\end{subsection}
\end{section}

\begin{section}{Main Estimates}
\label{section.main.estimates}
In the following section we shall provide the main estimates that will be used in the proofs appearing in Section \ref{section.proof.main.result}, namely Theorems \ref{theorem.quadratic.form.s.1}, \ref{theorem.quadratic.form}, and \ref{theorem.pauli.lieb.thirring.power.2}.

From now on we shall use the notation $\mathcal{P}_A$ for the Dirac operator $\sigma\cdot p_A$. The relevant Hilbert space where operators will act is $\mathcal{H} \equiv L^2(\mathbb{R}^3)\otimes \mathbb{C}^2$. We shall use the notation $D(A)$ and $Q(A)$ for the domain and the form domain of an operator $A$, respectively. In this section we will provide the proofs of Estimates I and II, which will allow us to effectively replace $|\mathcal{P}_A|^{2s}$ by $|p_A|^{2s}$ and $|p_A|^{2s}$ by $|p|^{2s}$, respectively. We remind the reader that $p = -i\nabla$ and $p_A = p - A$. We restrict for the moment to the case $0 < s < 1$. We shall utilize the word ``diamagnetism'' freely in what follows to refer to any inequality that derives from the pointwise estimate $|e^{-tp_A^2}\psi| \leq e^{-tp^2}|\psi|$, which is itself a simple consequence of the Feynman-Kac formula for the semigroup $e^{-tp_A^2}$.

We shall use the following interpolation repeatedly in what follows, and so we have found it convenient to state it as a lemma. It is an immediate consequence of H\"{o}lder's inequality, and consequently its proof will be omitted.
\begin{lemma}[General Interpolation Lemma]
\label{lemma.interpolation}
For a measure space $(\mathcal{M}, \mu)$ and a measurable function $f : \mathcal{M} \to \mathbb{C}$, if $\|f\|_a$ denotes the norm $\left(\int_{\mathcal{M}}|f|^a\,d\mu\right)^{1/a}$ for $a \geq 1$, we have that for numbers $1 \leq p \leq r \leq q \leq \infty$,
\begin{gather}
\|f\|_r \leq \|f\|_p^{(q - r)p/\left[(q - p)r\right]}\|f\|_q^{(r - p)q/\left[(q - p)r\right]}.
\end{gather}
\end{lemma}
In addition, we would like to record another lemma we will utilize, namely a generalization of Sobolev's inequality for the Pauli momentum $\mathcal{P}_A$.
\begin{lemma}[Pauli-Sobolev Inequality]
\label{lemma.pauli.sobolev}
For any $r > 3/2$, a locally $L^1$ function $A$ so that $\nabla\times A \equiv B$ is a function in $L^{r}(\mathbb{R}^3)$, and $0 < \varepsilon < 1$, we have that for any $\psi$ in the form domain of $\mathcal{P}_A^2$,
\begin{gather}
\|\psi\|_6^2 \leq S^2(1 - \varepsilon)^{-1}\||\mathcal{P}_A|\psi\|_2^2 + \omega(\varepsilon, r)\|B\|_{r}^{2r/(2r - 3)}\|\psi\|_2^2,
\end{gather}
where $S$ is the optimal constant in the Sobolev inequality $\|\psi\|_6 \leq S\||p|\psi\|_2$, equal to $(2/\pi)^{2/3}/\sqrt{3}$, and
\begin{gather}
\omega(\varepsilon, r) \equiv \frac{S^{4r/(2r - 3)}3^{3/(2r - 3)}(2r - 3)}{(1 - \varepsilon)(2r)^{2r/(2r - 3)}\varepsilon^{3/(2r - 3)}}.
\end{gather}
\end{lemma}
\begin{proof}
Assume initially that $\psi$ is in $C_0^{\infty}$. This additional assumption can be removed at the end by density. We first have that from the classical Sobolev inequality, diamagnetism, and the fact that $\mathcal{P}_A^2 = p_A^2 - \sigma\cdot B$,
\begin{gather}
\|\psi\|_6^2 = \||\psi|\|_6^2 \leq S^2(|\psi|, p^2|\psi|) \leq S^2(\psi, p_A^2\psi) = S^2(\psi, \mathcal{P}_A^2\psi) + S^2(\psi, \sigma\cdot B\psi).
\end{gather}
Since $(\sigma\cdot B)^2 = B^2$ and the Pauli matrices are Hermitian, we have that
\begin{gather}
|\sigma\cdot B \psi| = \sqrt{(\sigma\cdot B \psi, \sigma\cdot B \psi)_{\mathbb{C}^2}} = \sqrt{(\psi, (\sigma\cdot B)^2 \psi)_{\mathbb{C}^2}} = |B||\psi|
\end{gather}
(here $(a, b)_{\mathbb{C}^2} = \overline{a_1}b_1 + \overline{a_2}b_2$ is the inner product in $\mathbb{C}^2$), and therefore, by H\"{o}lder's inequality, Lemma \ref{lemma.interpolation}, and Young's inequality,
\begin{align}
S^2(\psi, \sigma\cdot B\psi) \leq \, & S^2|(\psi, \sigma\cdot B\psi)| \leq S^2\int_{\mathbb{R}^3}|B||\psi|^2\,dx\nonumber\\
\leq \, & S^2\|B\|_{p/(p - 1)}\|\psi\|_{2p}^2 \leq S^2\|B\|_{p/(p - 1)}\|\psi\|_2^{(3 - p)/p}\|\psi\|_6^{3(p - 1)/p}\nonumber\\
\leq \, & \frac{S^{4p/(3 - p)}(3 - p)}{2p}\delta^{-2p/(3 - p)}\|B\|_{p/(p - 1)}^{2p/(3 - p)}\|\psi\|_2^2 + \frac{3(p - 1)}{2p}\delta^{2p/3(p - 1)}\|\psi\|_6^2.
\end{align}
for any $\delta > 0$ and $1 \leq p < 3$. The result then follows from setting $\varepsilon \equiv 3(p - 1)\delta^{2p/3(p - 1)}/2p$ and $r \equiv p/(p - 1)$.
\end{proof}
We will now utilize these two lemmas to prove Estimate I in the fractional case $0 < s < 1$, providing an estimate on the difference of the operators $|\mathcal{P}_A|^{2s}$ and $|p_A|^{2s}$. In its statement certain functions depending on multiple variables appear. These are made explicit in an appendix at the end of the article. We state this and the second estimate with a mass $m \geq 0$; however, we never in this article make use of the massive case $m > 0$. The estimate here is written with a mass only because the comparison between $(|\mathcal{P}_A|^2 + m^2)^s - m^{2s}$ and $(|p_A|^2 + m^2)^s - m^{2s}$ adds no real extra effort over that between $|\mathcal{P}_A|^{2s}$ and $|p_A|^{2s}$.
\begin{theorem}[Estimate I: Fractional Case]
\label{theorem.estimate.I}
Let $m \geq 0$, and $s$, $u$, and $r$ be three numbers with the properties that $0 < s < 1$, $0 \leq u \leq 1$, $r > 3/2$, $3(1 - u) < 2r(1 - s)$. Let $A$ be a locally $L^1$ function such that $\nabla\times A = B$ is in $L^r(\mathbb{R}^3)$ and $\psi$ be in $Q(|\mathcal{P}_A|^{2s})\cap Q(|p_A|^{2s})$. In case $u = 0$,
\begin{gather}
\left|\left(\psi, \left(\left(\mathcal{P}_A^2 + m^2\right)^s - \left(p_A^2 + m^2\right)^s\right)\psi\right)\right| \leq \Theta(s, r)\|B\|_r^{2sr/(2r - 3)}\|\psi\|_2^2,
\end{gather}
where $\Theta$ is a positive function of $s$ and $r$. If now $u > 0$, the following two estimates hold: For any $\varepsilon > 0$,
\begin{align}
\left|\left(\psi, \left(\left(\mathcal{P}_A^2 + m^2\right)^s - \left(p_A^2 + m^2\right)^s\right)\psi\right)\right| \leq
\begin{cases}
\displaystyle\varepsilon\frac{\||p_A|^u\psi\|_2^{2s/u}}{\|\psi\|_2^{2(s - u)/u}} + \Omega(s, u, r, \varepsilon)\|B\|_r^{2sr/(2r - 3)}\|\psi\|_2^2,\vspace{3mm}\\
\displaystyle\varepsilon\frac{\||\mathcal{P}_A|^u\psi\|_2^{2s/u}}{\|\psi\|_2^{2(s - u)/u}} + \Omega(s, u, r, \varepsilon)\|B\|_r^{2sr/(2r - 3)}\|\psi\|_2^2,
\end{cases}
\end{align}
where $\Omega$ is a positive function of all the variables involved, divergent as $\varepsilon \to 0$.
\end{theorem}
\begin{proof}
As a means of moving through the proof with the least obstruction possible, we will assume that both $\psi$ and $A$ are $C_0^{\infty}$ functions and that $\|\psi\|_2 = 1$. At the end one can recover the general case by homogeneity and a density argument. For a given positive self-adjoint operator $A$ and any $0 < \alpha < 1$, the following identity holds,
\begin{gather}
A^{\alpha} = C_{\alpha}\int_0^{\infty}\frac{A}{(A + a)}\frac{\,da}{a^{1 - \alpha}},
\end{gather}
where $C_{\alpha}$ is equal to $\sin(\pi\alpha)/\pi$. By using the identity for both $\mathcal{P}_A^2 + m^2$ and $p_A^2 + m^2$ we obtain
\begin{align}
&\left(\psi,\left(\left(\mathcal{P}_A^2 + m^2\right)^s - \left(p_A^2 + m^2\right)^s\right)\psi\right)\nonumber\\
&\qquad \qquad = C_s \int_0^{\infty}\left(\psi, [(p_A^2 + m^2 + a)^{-1}-(\mathcal{P}_A^2 + m^2 + a)^{-1}]\psi \right) a^s\,da.
\label{equation.first.integral.estimate.I}
\end{align}
We split now the integral in \eqref{equation.first.integral.estimate.I} as $\int_0^{\alpha} + \int_{\alpha}^{\infty}$, where $\alpha$ is a parameter to be chosen optimal later. (The choice will be $\alpha \approx \| B \|_r^{2r/(2r-3)}$.) For small $a$ we estimate the norm of each resolvent as $\leq a^{-1}$. This yields the bound
\begin{align}\label{equation.small.alpha.part-new}
\left| C_s \int_0^{\alpha}
\left(\psi, [(p_A^2 + m^2 + a)^{-1}-(\mathcal{P}_A^2 + m^2 + a)^{-1}]\psi \right) a^s \,da\right|
\leq \frac{2 C_s \alpha^s}{s}.
\end{align}
In the remaining part of the integral, we apply the resolvent identity to get
\begin{align}
C_s &\int_{\alpha}^{\infty}
\left(\psi, [(p_A^2 + m^2 + a)^{-1}-(\mathcal{P}_A^2 + m^2 + a)^{-1}]\psi \right) a^s \,da\nonumber \\
= & \, -C_s\int_{\alpha} ^{\infty}\left(\psi, (\mathcal{P}_A^2 + m^2 + a)^{-1}(\sigma\cdot B)(p_A^2 + m^2 + a)^{-1}\psi\right) a^{s}\,da.
\label{equation.first.estimate.difference}
\end{align}
Here we proceed to estimate the integrand in \eqref{equation.first.estimate.difference}. By means of H\"{o}lder's inequality,
\begin{align}
& \left|\left(\psi, (\mathcal{P}_A^2 + m^2 + a)^{-1}\sigma\cdot B(p_A^2 + m^2 + a)^{-1}\psi\right)\right| = \left|\left((\mathcal{P}_A^2 + m^2 + a)^{-1}\psi, \sigma\cdot B (p_A^2 + m^2 + a)^{-1}\psi\right)\right|\nonumber\\
\leq & \, \int_{\mathbb{R}^3}\left|(\mathcal{P}_A^2 + m^2 + a)^{-1}\psi(x)\right|\left|(p_A^2 + m^2 + a)^{-1}\psi(x)\right||B(x)|\,dx\nonumber\\
\leq & \, \|B\|_r\|(\mathcal{P}_A^2 + m^2 + a)^{-1}\psi\|_p\|(p_A^2 + m^2 + a)^{-1}\psi\|_p,
\label{equation.estimate.difference}
\end{align}
where $2/p + 1/r = 1$ and $2 \leq p < 6$. (The reason for this last limitation is that we want to use Lemmas \ref{lemma.interpolation} and \ref{lemma.pauli.sobolev}.) We then have, using the Interpolation Lemma \ref{lemma.interpolation},
\begin{align}
& \|(\mathcal{P}_A^2 + m^2 + a)^{-1}\psi\|_p\|(p_A^2 + m^2 + a)^{-1}\psi\|_p\nonumber\\
\leq & \, \|(\mathcal{P}_A^2 + m^2 + a)^{-1}\psi\|_2^{(6 - p)/2p}\|(\mathcal{P}_A^2 + m^2 + a)^{-1}\psi\|_6^{3(p - 2)/2p}\nonumber\\
& \times\|(p_A^2 + m^2 + a)^{-1}\psi\|_2^{(6 - p)/2p}\|(p_A^2 + m^2 + a)^{-1}\psi\|_6^{3(p - 2)/2p}\nonumber\\
\leq & \, a^{3/r - 2}\|(\mathcal{P}_A^2 + m^2 + a)^{-1}\psi\|_6^{3/2r}\|(p_A^2 + m^2 + a)^{-1}\psi\|_6^{3/2r},
\end{align}
where in the last line we used that $2/p + 1/r = 1$ and that $\|(\mathcal{P}_A^2 + m^2 + a)^{-1}\psi\|_2 \leq \|(\mathcal{P}_A^2 + m^2 + a)^{-1}\|\|\psi\|_2 = (a + m^2)^{-1} \leq a^{-1}$, and similarly for $\|(a^2 + p_A^2)^{-1}\psi\|_2$. We now bound the remaining terms as follows. By means of Sobolev's inequality and diamagnetism we have
\begin{align}
& \left\Vert(p_A^2 + m^2 + a)^{-1}\psi\right\Vert_6 = \left\Vert\left|(p_A^2 + m^2 + a)^{-1}\psi\right|\right\Vert_6 \leq S\left\Vert |p|\left|(p_A^2 + m^2 + a)^{-1}\psi\right|\right\Vert_2\nonumber\\
\leq & \, S\left\Vert|p_A|(p_A^2 + m^2 + a)^{-1}\psi\right\Vert_2 = S\left\Vert|p_A|^{1 - u}(p_A^2 + m^2 + a)^{-1}|p_A|^u\psi\right\Vert_2\nonumber\\
\leq & \, S\left\Vert|p_A|^{1 - u}(p_A^2 + m^2 + a)^{-1}\right\Vert\left\Vert|p_A|^u\psi\right\Vert_2 \leq a^{-(1 + u)/2}D_u S\left\Vert |p_A|^u\psi\right\Vert_2,
\label{equation.estimate.sobolev}
\end{align}
where $D_u$ is the maximum of $x^{1 - u}/(x^2 + 1)$ on $[0, \infty)$ (equal to $2^{-1}(1 - u)^{(1 - u)/2}(1 + u)^{(1 + u)/2}$). Regarding the term involving $\mathcal{P}_A$, by means of the Pauli-Sobolev Inequality, Lemma \ref{lemma.pauli.sobolev}, we get that for any $0 < \delta < 1$,
\begin{align}
& \, \left\Vert (\mathcal{P}_A^2 + m^2 + a)^{-1}\psi\right\Vert_6\nonumber\\
\leq & \, \left(S^2(1 - \delta)^{-1}\left\Vert|\mathcal{P}_A|(\mathcal{P}_A^2 + m^2 + a)^{-1}\psi\right\Vert_2^2 + \omega(\delta, r)\left\Vert B\right\Vert_{r}^{2r/(2r - 3)}\|(\mathcal{P}_A^2 + m^2 + a)^{-1}\psi\|_2^2\right)^{1/2}\nonumber\\
\leq & \, \left(a^{-(1 + u)}S^2(1 - \delta)^{-1}D_u^2\||\mathcal{P}_A|^u\psi\|_2^2 + a^{-2}\omega(\delta, r)\left\Vert B\right\Vert_{r}^{2r/(2r - 3)}\right)^{1/2},
\end{align}
where an estimate similar to \eqref{equation.estimate.sobolev} was made. Since an integral will have to be performed at the end, we further bound this from above using the inequality $(x + y)^{\gamma} \leq x^{\gamma} + y^{\gamma}$ ($0 < \gamma < 1, x > 0, y >0$), obtaining
\begin{gather}
\left\Vert (\mathcal{P}_A^2 + m^2 + a)^{-1}\psi\right\Vert_6^{3/2r}\leq
E\left\Vert|\mathcal{P}_A|^u\psi\right\Vert_2^{3/2r}a^{-3(1 + u)/4r} + F\|B\|_r^{3/2(2r - 3)}a^{-3/2r},
\end{gather}
where $E$ and $F$ are certain constants. By assembling the estimates performed so far we find that, for some constants $G$ and $H$,
\begin{align}\label{equation.almost.final.estimate}
& \left|\left(\psi, (a^2 + m^2+ \mathcal{P}_A^2)^{-1}\sigma\cdot B(a^2 + m^2+ p_A^2)^{-1}\psi\right)\right|\nonumber\\
\leq & \, \left( G\|B\|_r\||\mathcal{P}_A|^u\psi\|_2^{3/2r}a^{-2 + 3(1 - u)/2r} + H\|B\|_r^{(4r - 3)/2(2r - 3)}a^{-2 + 3(1 - u)/4r} \right)\||p_A|^u\psi\|_2^{3/2r}.
\end{align}
We then use inequality \eqref{equation.almost.final.estimate} to bound $C_s\left|\int_{\alpha}^{\infty}\left(\psi, (\mathcal{P}_A^2 + m^2 + a)^{-1}\sigma\cdot B(p_A^2 + m^2 + a)^{-1}\psi\right)a^s\,da\right|$.
This we combine with \eqref{equation.small.alpha.part-new},
finally finding, after performing the change of variables $\alpha = \|B\|_r^{2r/(2r - 3)}\beta$,
\begin{align}
& \left|\left(\psi, \left(\left(\mathcal{P}_A^2 + m^2\right)^s - \left(p_A^2 + m^2\right)^s\right)\psi\right)\right|\nonumber\\
\leq & \, \frac{C_{s}}{s}\|B\|_r^{2sr/(2r - 3)}\beta^{s}\nonumber\\
& \, + I\|B\|_r^{(2sr - 3u)/(2r - 3)}\||\mathcal{P}_A|^u\psi\|_2^{3/2r}\||p_A|^u\psi\|_2^{3/2r}\beta^{-1 + s + 3(1 - u)/2r}\nonumber\\
& \, + J\|B\|_r^{(4sr - 3u)/2(2r - 3)}\||p_A|^u\psi\|_2^{3/2r}\beta^{-1 + s + 3(1 - u)/4r}.
\label{inequality.general.form.1}
\end{align}
In the case we can set $u = 0$, the right side of inequality \eqref{inequality.general.form.1} reduces to
\begin{gather}
\left(\frac{C_s}{s}\beta^s + I\beta^{s + 3/2r - 1} + J\beta^{s + 3/4r - 1}\right)\|B\|_r^{2sr/(2r - 3)},
\end{gather}
which is the first part of the theorem. In the case $u > 0$, by using Young's inequality for the second and third terms in \eqref{inequality.general.form.1}, we find
\begin{gather}
\left|\left(\psi, \left(\left(\mathcal{P}_A^2 + m^2\right)^s - \left(p_A^2 + m^2\right)^s\right)\psi\right)\right| \leq \eta\left(\||\mathcal{P}_A|^u\psi\|_2^{2s/u} + \||p_A|^u\psi\|_2^{2s/u}\right) + \Lambda\|B\|_r^{2sr/(2r - 3)},
\label{inequality.general.form.2}
\end{gather}
where $\Lambda$ is an explicit function of $\eta > 0$, $\delta$, $\beta$, $r$, $s$, and $u$. We now go one step further to eliminate either of $\mathcal{P}_A$ or $p_A$ in the right side of \eqref{inequality.general.form.2}. We shall show how to eliminate just one of them, in the understanding that eliminating the other involves a similar procedure. By using inequality \eqref{inequality.general.form.2} for $u = s$ and $m = 0$ we find
\begin{align}
& \||\mathcal{P}_A|^{u}\psi\|_2^2 = \left(\psi, |\mathcal{P}_A|^{2u}\psi\right) = \left(\psi, \left(|\mathcal{P}_A|^{2u} - |p_A|^{2u}\right)\psi\right) + \left(\psi, |p_A|^{2u}\psi\right)\nonumber\\
\leq \, & \gamma\||\mathcal{P}_A|^u\psi\|_2^2 + \left(\gamma + 1\right)\||p_A|^u\psi\|_2^2 + \Upsilon\|B\|_r^{2ur/(2r - 3)},
\end{align}
where we will choose $0 < \gamma < 1$ and $\Upsilon$ is the constant playing the role of $\Lambda$ (Equation \eqref{inequality.general.form.2}). By isolating $\mathcal{P}_A$ on the right and using the inequality $(x + y)^{\rho} \leq 2^{\rho - 1}(x^{\rho} + y^{\rho})$ ($x, y > 0$, $\rho \geq 1$), we obtain
\begin{gather}
\||\mathcal{P}_A|^u\psi\|_2^{2s/u} \leq 2^{s/u - 1}\left(\frac{1 + \gamma}{1 - \gamma}\right)^{s/u}\||p_A|^u\psi\|_2^{2s/u} + 2^{s/u - 1}\Upsilon^{s/u}\|B\|_r^{2sr/(2r - 3)},
\end{gather}
and this is how we find that
\begin{align}
\left|\left(\psi, \left(\left(\mathcal{P}_A^2 + m^2\right)^s - \left(p_A^2 + m^2\right)^s\right)\psi\right)\right| \leq &\, \eta\left(2^{s/u - 1}\left(\frac{1 + \gamma}{1 - \gamma}\right)^{s/u} + 1\right)\||p_A|^u\psi\|_2^{2s/u}\nonumber\\
& \, + \left(2^{s/u - 1}\Upsilon^{s/u}\eta + \Lambda\right)\|B\|_r^{2sr/(2r - 3)},
\end{align}
and the result then follows by calling
\begin{gather}
\varepsilon \equiv \eta\left(2^{s/u - 1}\left(\frac{1 + \gamma}{1 - \gamma}\right)^{s/u} + 1\right).
\end{gather}
\end{proof}
We will now provide an analogous theorem for the power $s = 1$. We shall need the following estimate before proceeding.
\begin{theorem}[Fractional Sobolev Inequality]
\label{theorem.fractional.sobolev}
For $N \geq 1$, $r \geq 2$, and tempered distributions $f$ on $\mathbb{R}^N$ such that their Fourier transform, $\widehat{f}$, is in $L_{\text{loc}}^1(\mathbb{R}^N)$, and $\int_{\mathbb{R}^N}\left|\xi\right|^{(1 - 2/r)N}|\widehat{f}(\xi)|^2\,d\xi < \infty$,
\begin{gather}
\|f\|_r \leq \mathcal{S}_{N, r}\||p|^{(1/2 - 1/r)N}f\|_{2},
\end{gather}
where
\begin{gather}
\mathcal{S}_{N, r} \equiv 2^{-N(1/2 - 1/r)}\pi^{-N(1/2 - 1/r)/2}\left[\frac{\Gamma(N/r)}{\Gamma[N(1 - 1/r)]}\right]^{1/2}\left[\frac{\Gamma(N)}{\Gamma(N/2)}\right]^{1/2 - 1/r}.
\end{gather}
\end{theorem}
The space where the distributions $f$ in Theorem \ref{theorem.fractional.sobolev} belong is known as the homogenous Sobolev space $\mathring{H}^{(1/2 - 1/r)N}(\mathbb{R}^N)$. The constant $\mathcal{S}_{N, r}$ is optimal. Theorem \ref{theorem.fractional.sobolev} appears in, for example, \cite[Equation (1.1)]{CFW}.
\begin{theorem}[Estimate I: The Case $s = 1$]
\label{theorem.estimate.I.s.1}
Let $\varepsilon > 0$, $u \in (0, 1]$, and $r \in [3/(2u), \infty)$. For functions $\psi \in L^2(\mathbb{R}^3)$ and locally $L^1(\mathbb{R}^3)$ functions $A$ with values in $\mathbb{R}^3$, we have that
\begin{gather}
\left|\left(\psi, \left(\mathcal{P}_A^2 - p_A^2\right)\psi\right)\right| \leq
\begin{cases}
\displaystyle\varepsilon\frac{\||p|^u\left|\psi\right|\|_2^{2/u}}{\|\psi\|_2^{2(1 - u)/u}} + \mathcal{T}(u, r, \varepsilon)\|B\|_r^{2r/(2r - 3)}\|\psi\|_2^2,\vspace{3mm}\\
\displaystyle\|B\|_{\infty}\|\psi\|_2^2,
\end{cases}
\end{gather}
where $B \equiv \nabla \times A$,
\begin{gather}
\mathcal{T}(u, r, \varepsilon) \equiv \frac{2r - 3}{2r}\left(\frac{3}{2r\varepsilon}\right)^{3/(2r - 3)}\mathcal{S}_{3, 6/(3 - 2u)}^{6/[u(2r - 3)]},
\end{gather}
and $\mathcal{S}$ is the constant from the Fractional Sobolev Inequality, Theorem \ref{theorem.fractional.sobolev}.
\end{theorem}
\begin{proof}
As before, we assume that $\|\psi\|_2 = 1$ and that $\psi \in C_0^{\infty}$. Let $1 < a \leq 3$. By noting that $\mathcal{P}_A^2 = p_A^2 - \sigma\cdot B$, by H\"{o}lder's inequality, with $1 \leq p \leq a$, the Interpolation Lemma \ref{lemma.interpolation}, and Theorem \ref{theorem.fractional.sobolev}, we find that
\begin{align}
& \left|\left(\psi, \left(\mathcal{P}_A^2 - p_A^2\right)\psi\right)\right| \leq \int_{\mathbb{R}^3}\left|\psi\right|^2|B|\,dx \leq \|\psi\|_{2p}^2\|B\|_{p/(p - 1)} \leq \|\psi\|_{2a}^{(p - 1)2a/\left[(a - 1)p\right]}\|B\|_{p/(p - 1)}\nonumber\\
\leq & \, \mathcal{S}_{3, 2a}^{(p - 1)2a/\left[(a - 1)p\right]}\||p|^{3/2 - 3/(2a)}\left|\psi\right|\|_{2}^{(p - 1)2a/\left[(a - 1)p\right]}\|B\|_{p/(p - 1)}.
\label{inequality.first.estimate.proof.estimate.I.case.s.1}
\end{align}
We now call $u \equiv 3/2 - 3/(2a)$, a number in $(0, 1]$, and $r \equiv p/(p - 1)$, which is in $[3/(2u), \infty]$, and obtain, by means of Young's inequality,
\begin{gather}
\mathcal{S}_{3, 6/(3 - 2u)}^{3/\left(ur\right)}\||p|^u\left|\psi\right|\|_{2}^{3/\left(ur\right)}\|B\|_{r} \leq \varepsilon\||p|^u\left|\psi\right|\|_{2}^{2/u} + \frac{2r - 3}{2r}\left(\frac{3}{2r\varepsilon}\right)^{3/(2r - 3)}\mathcal{S}_{3, 6/(3 - 2u)}^{6/[u(2r - 3)]}\|B\|_r^{2r/(2r - 3)},
\end{gather}
which proves the theorem for $3/(2u) \leq r < \infty$. The case $r = \infty$ can be obtained directly from the second inequality in \eqref{inequality.first.estimate.proof.estimate.I.case.s.1}.
\end{proof}
We shall concentrate now on providing an estimate for the absolute value of the expectation of the difference between $(p_A^2 + m^2)^s - m^{2s}$ and $(p^2 + m^2)^s - m^{2s}$ in the fractional case $0 < s < 1$, in complete analogy with the first estimate. It will be called ``Estimate II.'' Calculations will be similar to those encountered in the previous theorem, and therefore some steps will be developed in less detail than before. We will encounter again certain functions of several variables, whose explicit form we shall relegate to the appendix, as in Estimate I. We remark that Estimate II shall only be given in the fractional case, as the case $s = 1$ is not required in our work.
\begin{theorem}[Estimate II]
\label{theorem.estimate.II}
Let $s, u, r, m$ be four numbers with the properties that $0 < s < 1$, $0 \leq u \leq 1$, $3/2 < r < 3$, $m \geq 0$, and $3(1 - u) < 2r(1 - s)$. Let $A$ be a locally $L^1$ function that vanishes at infinity and such that $\nabla \times A \equiv B \in L^r$ and $\nabla\cdot A = 0$, and $\psi$ be a function in $Q(|p|^{2s})\cap Q(|p_A|^{2s})$. Then, if $u = 0$,
\begin{gather}
\left|\left(\psi, \left(\left(p_A^2 + m^2\right)^s - \left(p^2 + m^2\right)^s\right)\psi\right)\right| \leq \mathcal{I}(s, r)\left\Vert B\right\Vert_r^{2sr/(2r - 3)}\left\Vert\psi\right\Vert_2^2,
\end{gather}
where $\mathcal{I}$ is an explicit positive function of both $s$ and $r$. If now $u > 0$, the following two inequalities hold
\begin{gather}
\left|\left(\psi, \left(\left(p_A^2 + m^2\right)^s - \left(p^2 + m^2\right)^s\right)\psi\right)\right| \leq
\begin{cases}
\displaystyle\varepsilon\frac{\||p|^u\psi\|_2^{2s/u}}{\|\psi\|_2^{2(s - u)/u}} + \mathcal{J}(s, u, r, \varepsilon)\|B\|_r^{2sr/(2r - 3)}\|\psi\|_2^2\vspace{3mm},\\
\displaystyle\varepsilon\frac{\||p_A|^u\psi\|_2^{2s/u}}{\|\psi\|_2^{2(s - u)/u}} + \mathcal{J}(s, u, r, \varepsilon)\|B\|_r^{2sr/(2r - 3)}\|\psi\|_2^2,
\end{cases}
\end{gather}
where $\mathcal{J}$ is an explicit positive function of $s$, $u$, $r$, and also of the additional variable $\varepsilon >0$. $\mathcal{J}$ diverges as $\varepsilon \to 0$.
\end{theorem}
\begin{proof}
For simplicity we assume as before that both $A$ and $\psi$ are $C_0^{\infty}$ functions and that $\|\psi\|_2 = 1$; the general case follows by replacing $\psi$ by $\psi/\|\psi\|_2$ and density. After performing a resolvent expansion for the difference between the operators in question, namely $(p_A^2 + m^2)^s$ and $(p^2 + m^2)^s$, and using the fact that $p_A^2 - p^2 = -p_A\cdot A - A\cdot p$, we find, by essentially the same arguments as in the previous theorem,
\begin{align}
& \left|\left(\psi, \left(\left(p_A^2 + m^2\right)^s - \left(p^2 + m^2\right)^s\right)\psi\right)\right|\nonumber\\
 \leq &\, \frac{C_s \alpha^s}{s} + C_s\int_{\alpha}^{\infty}\left|\left(\psi, \left(p_A^2 + m^2 + a\right)^{-1}(p_A\cdot A)\left(p^2 + m^2 + a\right)^{-1}\psi\right)\right|a^s\,da\nonumber\\
& \, + C_s\int_{\alpha}^{\infty}\left|\left(\psi, \left(p_A^2 + m^2 + a\right)^{-1}(A\cdot p)\left(p^2 + m^2 + a\right)^{-1}\psi\right)\right|a^s\,da,
\label{second.section.last.term.first.estimate}
\end{align}
where $\alpha > 0$ is a constant to be fixed later, and $C_s$ has the same meaning as in the previous proof. The absolute value of the expectation in the second term above may be bounded above as
\begin{align}
&\left|\left(\psi, \left(p_A^2 + m^2 + a\right)^{-1}(p_A\cdot A)\left(p^2 + m^2 + a\right)^{-1}\psi\right)\right|\nonumber\\
&\qquad\qquad=  \left|\left(p_A\left(p_A^2 + m^2 + a\right)^{-1}\psi, A\left(p^2 + m^2 + a\right)^{-1}\psi\right)\right|\nonumber\\
&\qquad\qquad\leq  \, \left\Vert p_A\left(p_A^2 + m^2 + a\right)^{-1}\psi\right\Vert_2\left\Vert A\right\Vert_p\left\Vert (p^2 + m^2 + a)^{-1}\psi\right\Vert_q\nonumber\\
&\qquad\qquad=  \, \left\Vert |p_A|\left(p_A^2 + m^2 + a\right)^{-1}\psi\right\Vert_2\left\Vert A\right\Vert_p\left\Vert (p^2 + m^2 + a)^{-1}\psi\right\Vert_q,
\label{second.section.second.term}
\end{align}
where $2 < q < 6$ and $p^{-1} + q^{-1} = 1/2$. The first term in this last expression, \eqref{second.section.second.term}, has already been encountered in the previous proof -- by repeating what amounts to the same calculations we find
\begin{gather}
\left\Vert|p_A|(p_A^2 + m^2 + a)^{-1}\psi\right\Vert_2 \leq D_u\|A\|_p\|(p^2 + m^2 + a)^{-1}\psi\|_q\||p_A|^u\psi\|_2 a^{-(1 + u)/2}
\end{gather}
for any $0 \leq u \leq 1$, where $D_u$ is as in the previous proof. The last term in \eqref{second.section.second.term} has also already been dealt with in the last proof, and by repeating those arguments we find
\begin{gather}
\|(p^2 + m^2 + a)^{-1}\psi\|_q \leq S^{3(q - 2)/2q}D_u^{3(q - 2)/2q}\||p|^u\psi\|_2^{3(q - 2)/2q}a^{-3(1 - u)/q - (1 + 3u)/4},
\end{gather}
where again $S$ is the constant in the classical Sobolev inequality (see Lemma \ref{lemma.pauli.sobolev}). Since $2 < q < 6$, we have that $3 < p < \infty$, and therefore we may use the general Sobolev inequality for rotors alluded to above (see the appendix),
\begin{gather}
\|A\|_p \leq N_{3p/(3 + p)}\|B\|_{3p/(3 + p)}.
\end{gather}
Furthermore, we would like to express everything in terms of a single variable, namely $3p/(3 + p)$, that we shall call $r$. In this way we obtain that the entire term \eqref{second.section.second.term} is bounded above by
\begin{gather}
D_u^{3/r}S^{(3 - r)/r}N_r\left\Vert|p_A|^u\psi\right\Vert_2\left\Vert|p|^u\psi\right\Vert_2^{(3 - r)/r}\left\Vert B\right\Vert_r a^{-(5r - 6)(1 - u)/4r - (3 + 5u)/4}
\label{equation.estimate.I.first.bound}
\end{gather}
for all $3/2 < r < 3$. In a similar way, as regards the last term in \eqref{second.section.last.term.first.estimate}, the absolute value appearing there can be bounded as
\begin{align}
&\left|\left(\psi, \left(p_A^2 + m^2 + a\right)^{-1}(A\cdot p)\left(p^2 + m^2 + a\right)^{-1}\psi\right)\right|\nonumber\\
= & \left|\left(\left(p_A^2 + m^2 + a\right)^{-1}\psi, \left(A\cdot p\right)\left(p^2 + m^2 + a\right)^{-1}\psi\right)\right|\nonumber\\
\leq & \, \left\Vert p(p^2 + m^2 + a)^{-1}\psi\right\Vert_2\left\Vert A\right\Vert_p\left\Vert\left(p_A^2 + m^2 + a\right)^{-1}\psi\right\Vert_q\nonumber\\
= & \, \left\Vert |p|(p^2 + m^2 + a)^{-1}\psi\right\Vert_2\left\Vert A\right\Vert_p\left\Vert\left(p_A^2 + m^2 + a\right)^{-1}\psi\right\Vert_q,
\end{align}
and therefore essentially the same computations as just done can be performed, and so we find that this is further bounded by
\begin{gather}
D_u^{3/r}S^{(3 - r)/r}N_r\left\Vert|p|^u\psi\right\Vert_2\left\Vert|p_A|^u\psi\right\Vert_2^{(3 - r)/r}\left\Vert B\right\Vert_r a^{-(5r - 6)(1 - u)/4r - (3 + 5u)/4}.
\label{equation.estimate.I.second.bound}
\end{gather}
By using then the estimates \eqref{equation.estimate.I.first.bound} and \eqref{equation.estimate.I.second.bound} in \eqref{second.section.last.term.first.estimate} and optimizing over $\alpha$ we find
\begin{align}
\left|\left(\psi, \left(\left(p_A^2 + m^2\right)^s - \left(p^2 + m^2\right)^s\right)\psi\right)\right| \leq & \, \mathcal{E}_{s, u, r}\left(\||p|^u\psi\|_2\||p_A|^u\psi\|_2\right)^{2s(3 - r)/\left[2r - 3(1 - u)\right]}\nonumber\\
&\times\left(\||p_A|^u\psi\|_2^{(2r - 3)/r} + \||p|^u\psi\|_2^{(2r - 3)/r}\right)^{2rs/\left[2r - 3(1 - u)\right]}\nonumber\\
&\times\|B\|_r^{2rs/\left[2r - 3(1 - u)\right]},
\end{align}
where $\mathcal{E}$ is an explicit function of $s, u, r$, which will not be written out, so as not to unnecessarily obscure the main ideas in the proof. By setting $u = 0$ we then obtain the first result of the theorem. If $u > 0$, we obtain, by first splitting the term in the middle using $(x + y)^{\gamma} \leq x^{\gamma} + y^{\gamma}$ for $x, y > 0$, $0 < \gamma < 1$, and then using Young's inequality,
\begin{gather}
\left|\left(\psi, \left(\left(p_A^2 + m^2\right)^s - \left(p^2 + m^2\right)^s\right)\psi\right)\right| \leq \xi\left(\||p|^u\psi\|_2^{2s/u} + \||p_A|^u\psi\|_2^{2s/u}\right) + \mathcal{F}_{s, u, r, \xi}\|B\|_r^{2sr/(2r - 3)},
\end{gather}
where $\mathcal{F}$ is again an explicit function of $s, u, r$ and $\xi > 0$. By eliminating either $|p|$ or $|p_A|$, in exactly the same way as was done in the previous theorem, the final result is obtained, where a new variable $0 < \delta < 1$ is introduced and $\varepsilon$ is defined as $\xi\left[1 + 2^{s/u - 1}\delta^{s/u}(1 - \delta)^{-s/u}\right]$.
\end{proof}
Having finished with the proofs of Estimates I and II, we shall provide now some consequences of the estimates that will be useful in the proofs that will follow in Section \ref{section.proof.main.result}.
\begin{theorem}[Quadratic Form Estimate for $s = 1$]
\label{theorem.quadratic.form.s.1}
For each $\lambda \in [0, 1)$ and $3/2 < r \leq \infty$ there exists a constant $C > 0$ such that
\begin{gather}
|\mathcal{P}_A|^2 - \frac{C_1}{|x|^2} \geq \lambda\left(|p_A|^2 - \frac{C_1}{|x|^2}\right) - C\|B\|_r^{2r/(2r - 3)}.
\label{inequality.quadratic.form.s.1}
\end{gather}
$\lambda$ may be set equal to 1 if $r = \infty$, yielding a finite value of $C$.
\end{theorem}
\begin{proof}
This follows easily from Theorem \ref{theorem.estimate.I.s.1}, Equation \eqref{inequality.SSS}, and diamagnetism, since for a normalized $C_0^{\infty}$ function $\psi$ and $\varepsilon > 0$ small enough,
\begin{align}
\left(\psi, \left(|\mathcal{P}_A|^2 - \frac{C_1}{|x|^2}\right)\psi\right) & \geq \, \left(\psi, \left(|p_A|^2 - \frac{C_1}{|x|^2}\right)\psi\right) - \varepsilon\||p|^u\left|\psi\right|\|_2^{2/u} - \mathcal{T}\|B\|_r^{2r/(2r -3)}\nonumber\\
& \geq \, \lambda\left(\psi, \left(|p_A|^2 - \frac{C_1}{|x|^2}\right)\psi\right) + (1 - \lambda)\left(\left|\psi\right|, \left(|p|^2 - \frac{C_1}{|x|^2}\right)\left|\psi\right|\right)\nonumber\\
& \,\, \quad - \varepsilon\||p|^u\left|\psi\right|\|_2^{2/u} - \mathcal{T}\|B\|_r^{2r/(2r -3)}\nonumber\\
& \geq \, \lambda\left(\psi, \left(|p_A|^2 - \frac{C_1}{|x|^2}\right)\psi\right) - \mathcal{T}\|B\|_r^{2r/(2r - 3)},
\end{align}
as claimed. The situation with $r = \infty$ is easy to verify.
\end{proof}
\begin{theorem}[Quadratic Form Estimates for Fractional Powers]
\label{theorem.quadratic.form}
Let $0 < s < 1$ and $\lambda \in [0, 1)$. There are constants $D(s, \lambda), E(s, \lambda) > 0$ such that
\begin{align}
|\sigma\cdot p_A|^{2s} - \frac{C_s}{|x|^{2s}} \geq & \, \lambda\left(|p|^{2s} - \frac{C_s}{|x|^{2s}}\right) - D(s, \lambda)\int_{\mathbb{R}^3}|B|^{s + 3/2}\,dx,\label{inequality.quadratic.forms.first}\\
|\sigma\cdot p_A|^{2s} - \frac{C_s}{|x|^{2s}} \geq & \, \lambda\left(|p|^{2s} - \frac{C_s}{|x|^{2s}}\right) - E(s, \lambda)\left(\int_{\mathbb{R}^3}|B|^2\,dx\right)^{2s}.\label{inequality.quadratic.forms.second}
\end{align}
\end{theorem}
\begin{proof}
We will only prove \eqref{inequality.quadratic.forms.first}. The proof of \eqref{inequality.quadratic.forms.second} is basically identical to the one we will give now, except that the power of 2 must be selected for the magnetic field. The proof is a careful application of Estimates I and II. From Estimate I it follows that
\begin{gather}
\left(\psi, \left(|\mathcal{P}_A|^{2s} - \frac{C_s}{|x|^{2s}}\right)\psi\right) \geq \left(\psi, \left(|p_A|^{2s} - \frac{C_s}{|x|^{2s}}\right)\psi\right) - \varepsilon\||p_A|^u\psi\|_2^{2s/u} - \Omega\int_{\mathbb{R}^3}|B|^{s + 3/2}\,dx
\label{phlt.inequality.1}
\end{gather}
for any $C_0^{\infty}$ function $\psi$ of norm 1. Here, we recall, $\varepsilon$ is any positive number, and $\Omega$ is a function of all the variables involved, that is, $s$, $u$, and $\varepsilon$. Inequality \eqref{phlt.inequality.1} is valid as long as $s(2s + 1)/3 < u < s$, and so we assume that is the case. (See the statement of Estimate I.) We now apply Estimate II to the first two terms on the right hand side of \eqref{phlt.inequality.1}. For the first one we obtain an expression analogous to \eqref{phlt.inequality.1}, namely
\begin{gather}
\left(\psi, \left(|p_A|^{2s} - \frac{C_s}{|x|^{2s}}\right)\psi\right) \geq \left(\psi, \left(|p|^{2s} - \frac{C_s}{|x|^{2s}}\right)\psi\right) - \varepsilon\||p|^u\psi\|_2^{2s/u} - \mathcal{J}\int_{\mathbb{R}^3}|B|^{s + 3/2}\,dx,
\label{phlt.inequality.2}
\end{gather}
where again $\mathcal{J}$ depends on $s$, $u$, and $\varepsilon$. The parameters $u$ and $\varepsilon$ are here chosen to be equal to the values appearing in \eqref{phlt.inequality.1}. Note how the first term on the right side of \eqref{phlt.inequality.2} controls the second one for $\varepsilon$ sufficiently small, by the equation \eqref{inequality.SSS}. It is then the second term on the right side of \eqref{phlt.inequality.1} the one that is left to control. We simply apply Estimate II to it, more specifically,
\begin{align}
\||p_A|^u\psi\|_2^{2s/u} = \left(\psi, |p_A|^{2u}\psi\right)^{s/u} \leq & \, \left[\left(1 + \varepsilon\right)\left(\psi, |p|^{2u}\psi\right) + \mathcal{K}\left(\int_{\mathbb{R}^3}|B|^{s + 3/2}\,dx\right)^{u/s}\right]^{s/u}\nonumber\\
\leq & \, 2^{s/u - 1}(1 + \varepsilon)^{s/u}\||p|^u\psi\|^{2s/u} + 2^{s/u - 1}\mathcal{K}^{s/u}\int_{\mathbb{R}^3}|B|^{s + 3/2}\,dx,
\label{phlt.inequality.3}
\end{align}
where the ``$u$-parameter'' was picked equal to the ``$s$-parameter,'' here equal to $u$. (See the statement of Estimate II.) $\mathcal{K}$ is here a constant, in complete analogy to the constant $\mathcal{J}$ we saw before. By inserting inequalities \eqref{phlt.inequality.3} and \eqref{phlt.inequality.2} into \eqref{phlt.inequality.1} we obtain
\begin{align}
& \left(\psi, \left(|\mathcal{P}_A|^{2s} - \frac{C_s}{|x|^{2s}}\right)\psi\right)\nonumber\\
\geq & \, \left(\psi, \left(|p|^{2s} - \frac{C_s}{|x|^{2s}}\right)\psi\right) - \varepsilon\left[1 + 2^{s/u - 1}\left(1 + \varepsilon\right)^{s/u}\right]\||p|^u\psi\|_2^{2s/u} - \mathcal{L}\int_{\mathbb{R}^3}|B|^{s + 3/2}\,dx,
\end{align}
where $\mathcal{L} \equiv 2^{s/u - 1}\mathcal{K}^{s/u}\varepsilon + \mathcal{J} + \Omega$. By picking $\varepsilon$ small enough and using \eqref{inequality.SSS} we finally find that
\begin{gather}
|\mathcal{P}_A|^{2s} - \frac{C_s}{|x|^{2s}} \geq \sigma\left(|p|^{2s} - \frac{C_s}{|x|^{2s}}\right) - \mathcal{L}\int_{\mathbb{R}^3}|B|^{s + 3/2}\,dx,
\label{phlt.main.estimate.near.origin}
\end{gather}
where $\sigma$ is any number in the interval $[0, 1)$. ($\mathcal{L}$ is a function of $\sigma$ through $\varepsilon$; it blows up as $\sigma \to 1^{-}$.)
\end{proof}
We close this section by giving a preliminary result concerning our Pauli-Hardy-Lieb-Thirring inequality in terms of the magnetic field energy $\|B\|_2^2$, Theorem \ref{theorem.main.result}, that will be used in its proof. It is the generalization to powers $1/2 \leq s \leq 1$ of a result of Lieb, Loss, and Solovej \cite{LLS}.
\begin{theorem}[Pauli-Lieb-Thirring Inequality with Magnetic Field Energy]
\label{theorem.pauli.lieb.thirring.power.2}
For $1/2 \leq s \leq 1$ and $0 < \gamma < 1$,
\begin{gather}
\text{\textnormal{Tr}}(|\mathcal{P}_A|^{2s} + W)_{-} \leq U(s, \gamma) \int_{\mathbb{R}^3}W_{-}^{1 + 3/(2s)}\,dx + V(s, \gamma)\left(\int_{\mathbb{R}^3}|B|^2\,dx\right)^{3/4}\left(\int_{\mathbb{R}^3}W_{-}^4\,dx\right)^{1/4},
\end{gather}
where
\begin{align}
U(s, \gamma) \equiv \, & \sqrt{2}Ls(1 - \gamma)^{-s}\Omega(s, 0),\\
V(s, \gamma) \equiv \, & 4\cdot 3^{-3/4}\sqrt{2}Ls\Omega(s, 2s - 1)^{3/4}\Omega(4s - 3/2, 0)^{1/4}(1 - \gamma)^{s - 3/8}\gamma^{-3/8},\\
\Omega(s, r) \equiv \, & \int_0^1\lambda^{s - 1}(1 - \lambda^{r + 1})^{3/2}\,d\lambda,
\end{align}
for $1/2 < s \leq 1$, and
\begin{gather}
U(1/2, \gamma) = \frac{3\pi L}{2},\\
V(1/2, \gamma) = \frac{\pi L}{2\cdot 3^{1/4}}.
\end{gather}
(Independent of $\gamma$.) $L$ is here is the constant $L_3$ in the CLR bound \eqref{equation.CLR}, smaller than 0.1156.
\end{theorem}
\begin{proof}[Proof of Theorem \ref{theorem.pauli.lieb.thirring.power.2}]
The case $s = 1$ of the inequality we shall prove has already been covered in a paper by Lieb, Loss and Solovej \cite[Theorem 2]{LLS}. The idea of the following proof is simply to modify their ``running-energy-scale'' argument to accomodate the different powers $1/2 < s < 1$ of $\mathcal{P}_A^2$. We shall explain the case $s = 1/2$ at the end. We start by first using the trivial bound $W \geq -W_{-}$ and the BKS inequality, Equation \eqref{inequality.BKS}:
\begin{gather}
\text{Tr}\, (|\mathcal{P}_A|^{2s} + W)_{-} \leq \text{Tr}\,(|\mathcal{P}_A|^{2s} - W_{-})_{-} \leq \text{Tr}\,(\mathcal{P}_A^2 - W_{-}^{1/s})_{-}^s.
\label{equation.lt.inequality.pauli.first}
\end{gather}
The rightmost element in \eqref{equation.lt.inequality.pauli.first} can be written as
\begin{gather}
s\int_0^{\infty}\lambda^{s - 1}\mathcal{N}(\mathcal{P}_A^2 - W_{-}^{1/s} + \lambda)\,d\lambda,
\label{equation.lt.inequality.trace}
\end{gather}
where $\mathcal{N}(X)$ denotes the number of non-positive eigenvalues of an operator $X$. (See \cite[Equation (9)]{LA}.) We fix an energy scale $\mu > 0$, to be selected later, and split \eqref{equation.lt.inequality.trace} as
\begin{gather}
s\int_0^{\mu}\lambda^{s - 1}\mathcal{N}(\mathcal{P}_A^2 - W_{-}^{1/s} + \lambda)\,d\lambda + s\int_{\mu}^{\infty}\lambda^{s - 1}\mathcal{N}(\mathcal{P}_A^2 - W_{-}^{1/s} + \lambda)\,d\lambda \equiv I_1 + I_2.
\end{gather}
For the first integral we can use the elementary estimate $\mathcal{P}_A^2 \geq p_A^2 - |B|$, and then the Cwikel-Lieb-Rozenblum bound, Equation \eqref{equation.CLR},
\begin{gather}
\mathcal{N}\left[p_A^2 + V\right] \leq L\int_{\mathbb{R}^3}V_{-}^{3/2}\,dx,
\end{gather}
with the constant $L$ defined as $0.1156$, obtaining in particular
\begin{gather}
\mathcal{N}(\mathcal{P}_A^2 - W_{-}^{1/s} + \lambda) \leq \mathcal{N}(p_A^2 - |B| - W_{-}^{1/s} + \lambda) \leq L\int_{\mathbb{R}^3}\left(\lambda - |B| - W_{-}^{1/s}\right)_{-}^{3/2}\,dx.
\label{equation.lt.inequality.pauli.second}
\end{gather}
We furthermore fix $0 < \gamma < 1$ and $r \geq 0$ and bound the integrand on the rightmost end of \eqref{equation.lt.inequality.pauli.second} as
\begin{gather}
\left(\lambda - |B| - W_{-}^{1/s}\right)_{-}^{3/2} \leq \sqrt{2}\left[\left(\gamma\lambda\left(\lambda/\mu\right)^r - |B|\right)_{-}^{3/2} + \left((1 - \gamma)\lambda - W_{-}^{1/s}\right)_{-}^{3/2}\right].
\end{gather}
As for the second integral $I_2$, we bound from below the operator $\mathcal{P}_A^2$ as $\left(\mu \lambda^{-1}\right)^r \left[(p - A)^2 - |B|\right]$, and then
\begin{align}
\mathcal{N}\left(\mathcal{P}_A^2 - W_{-}^{1/s} + \lambda\right) & \leq \mathcal{N}\left[\left(\mu \lambda^{-1}\right)^{r}\left[p_A^2 - |B|\right] - W_{-}^{1/s} + \lambda\right]\nonumber\\
& = \mathcal{N}\left[p_A^2 - |B| - \left(\lambda\mu^{-1}\right)^r W_{-}^{1/s} + \lambda^{r + 1}/\mu^r \right]\nonumber\\
& \leq L\int_{\mathbb{R}^3}\left(\lambda^{r + 1}/\mu^r - |B| - (\lambda\mu^{-1})^r W_{-}^{1/s}\right)_{-}^{3/2}\,dx.\label{equation.lt.inequality.pauli.third}
\end{align}
From here we find, after splitting the integrand in \eqref{equation.lt.inequality.pauli.third} as
\begin{align}
\left(\lambda^{r + 1}/\mu^r - |B| - (\lambda\mu^{-1})^r W_{-}^{1/s}\right)_{-}^{3/2} \leq \sqrt{2} & \left[\left(\gamma\lambda^{r + 1}/\mu^r - |B|\right)_{-}^{3/2}\right.\nonumber\\
& \quad + \left.\left( (1 - \gamma)\lambda^{r + 1}/\mu^r - \left(\lambda\mu^{-1}\right)^r W_{-}^{1/s}\right)_{-}^{3/2}\right],
\end{align}
that
\begin{align}
I_1 + I_2 \leq \sqrt{2}Ls\int_{\mathbb{R}^3} & \left[\int_0^{\infty}\lambda^{s - 1}\left(\gamma\lambda^{r + 1}\mu^{-r} - |B|\right)_{-}^{3/2}\,d\lambda\right.\nonumber\\
& \quad + \int_0^{\mu}\lambda^{s - 1}\left((1 - \gamma)\lambda - W_{-}^{1/s}\right)_{-}^{3/2}\,d\lambda\nonumber\\
& \quad + \left.\int_{\mu}^{\infty}\lambda^{s - 1}\left((1 - \gamma)\lambda^{r + 1}\mu^{-r} - \lambda^r\mu^{-r}W_{-}^{1/s}\right)_{-}^{3/2}\,d\lambda\right]\,dx.
\end{align}
By now extending the last two integrals from 0 to $\infty$ we get, with
\begin{gather}
\Omega(s, r) \equiv \int_0^1\lambda^{s - 1}(1 - \lambda^{r + 1})^{3/2}\,d\lambda,
\end{gather}
that
\begin{align}
I_1 + I_2 \leq \sqrt{2}Ls & \left(\mu^{rs/(r + 1)}\gamma^{-s/(r + 1)}\Omega(s, r)\int_{\mathbb{R}^3}|B|^{3/2 + s/(r + 1)}\,dx\right.\nonumber\\
& \quad + \, (1 - \gamma)^{-s}\Omega(s, 0)\int_{\mathbb{R}^3}W_{-}^{1 + 3/(2s)}\,dx\nonumber\\
& \quad + \left. \mu^{-3r/2}(1 - \gamma)^{-(s + 3r/2)}\Omega(s + 3r/2, 0)\int_{\mathbb{R}^3}W_{-}^{1 + 3(r + 1)/(2s)}\,dx\right).
\end{align}
In order to get $\|B\|_2$ we must select $r = 2s - 1$. This can be done only for $s \geq 1/2$. In the case $s > 1/2$, we optimize over $\mu$ and obtain
\begin{gather}
\text{Tr}\,(|\mathcal{P}_A|^{2s} + W)_{-} \leq U(s, \gamma) \int_{\mathbb{R}^3}W_{-}^{1 + 3/(2s)}\,dx + V(s, \gamma)\left(\int_{\mathbb{R}^3}|B|^2\,dx\right)^{3/4}\left(\int_{\mathbb{R}^3}W_{-}^4\,dx\right)^{1/4},
\end{gather}
where
\begin{gather}
U(s, \gamma) \equiv \sqrt{2}Ls(1 - \gamma)^{-s}\Omega(s, 0),\\
V(s, \gamma) \equiv 4\cdot 3^{-3/4}\sqrt{2}Ls\Omega(s, 2s - 1)^{3/4}\Omega(4s - 3/2, 0)^{1/4}(1 - \gamma)^{s - 3/8}\gamma^{-3/8}.
\end{gather}

We will now discuss the case $s = 1/2$. This was actually proven in \cite[Section VI. Proof of Theorem 2.2]{EFS}, but for completeness we will prove it again here. (There is another reason besides completeness, namely that they did not say what $U$ and $V$ are.) As in the case $1/2 < s < 1$, the first step is to use the BKS inequality, obtaining
\begin{gather}
\text{Tr}\,(|\mathcal{P}_A| + W)_{-} \leq \text{Tr}\,(\mathcal{P}_A^2 - W_{-}^2)^{1/2}_{-},
\end{gather}
and then use the running-energy-scale method in a slightly different way than before; setting a number $0 \leq \mu \leq 1$, to be fixed later, we obtain, by remembering that $\mathcal{P}_A^2 \geq p_A^2 - |B|$ and using the CLR bound \eqref{equation.CLR},
\begin{align}
&\text{Tr}\,(\mathcal{P}_A^2 - W_{-}^2)^{1/2}_{-} = \frac{1}{2}\int_0^{\infty}\lambda^{-1/2}\mathcal{N}(\mathcal{P}_A^2 - W_{-}^2 + \lambda)\,d\lambda \leq \frac{1}{2}\int_0^{\infty}\lambda^{-1/2}\mathcal{N}(\mu\mathcal{P}_A^2 - W_{-}^2 + \lambda)\,d\lambda\nonumber\\
\leq \, & \frac{1}{2}\int_0^{\infty}\lambda^{-1/2}\mathcal{N}(\mu p_A^2 - \mu |B| - W_{-}^2 + \lambda)\,d\lambda = \frac{1}{2}\int_0^{\infty}\lambda^{-1/2}\mathcal{N}(p_A^2 - |B| - W_{-}^2/\mu + \lambda/\mu)\,d\lambda\nonumber\\
\leq \, & \frac{L}{2}\int_{\mathbb{R}^3}\int_0^{\infty}\lambda^{-1/2}\left(\lambda/\mu - W_{-}^2/\mu - |B|\right)_{-}^{3/2}\,d\lambda\,dx = \frac{3\pi L}{16}\mu^{1/2}\int_{\mathbb{R}^3}\left(W_{-}^2/\mu + |B|\right)^2\,dx\nonumber\\
\leq \, & \frac{3\pi L}{8}\left(\mu^{-3/2}\int_{\mathbb{R}^3}W_{-}^4\,dx + \mu^{1/2}\int_{\mathbb{R}^3}|B|^2\,dx\right) \equiv \frac{3\pi L}{8}\left(\omega^{-3}M + \omega N\right).
\end{align}
We see then that the optimal $\omega$ is
\begin{gather}
\min\left[1, \left(\frac{3M}{N}\right)^{1/4}\right],
\end{gather}
and therefore
\begin{align}
\omega^{-3}M + \omega N \leq & \,
\begin{cases}
3^{-3/4}\cdot 4M^{1/4}N^{3/4} & \text{ if $3M \leq N$}\\
4M & \text{ if $3M > N$},
\end{cases}\nonumber\\
\leq & \,\, 3^{-3/4}\cdot 4M^{1/4}N^{3/4} + 4M,
\end{align}
which concludes the proof.
\end{proof}
\end{section}

\begin{section}{The Main Localization and Estimates on the Trace}
One of the main tools in the proof of the Hardy-Lieb-Thirring inequality is the following localization estimate. It makes reference to a ``standard exponential localization,'' which we describe presently: we fix a length $b > 0$ and a base $l > 1$, and consider a partition $\phi_0, \phi_1, \phi_2, \ldots$ of the interval $[0, \infty)$ defined roughly as a chain of bumps $\phi_n$ of height 1 and length of order $bl^n$. From here we construct a partition of all of $\mathbb{R}^3$ into shells as $\varphi_n (x) = \phi_n(|x|)$. A precise definition of the functions $\phi_n$ is as follows:
\begin{gather}
\phi_0(t) = 1_{[0, b]} + \cos\left[\frac{\pi(t - b)}{2bl}\right]1_{(b, b + bl]},\\
\phi_1(t) = \sin\left[\frac{\pi(t - b)}{2bl}\right]1_{[b, b + bl]} + \cos\left[\frac{\pi(t - b - bl)}{2bl^2}\right]1_{(b + bl, b + bl + bl^2]},\\
\phi_2(t) = \sin\left[\frac{\pi(t - b - bl)}{2bl^2}\right]1_{[b + bl, b + bl + bl^2]}(t) + \cos\left[\frac{\pi(t - b - bl - bl^2)}{2bl^3}\right]1_{(b + bl + bl^2, b + bl + bl^2 + bl^3]},\\
\phi_n(t) = \sin\left[\frac{\pi(t - s_{n - 1})}{2bl^n}\right]1_{[s_{n - 1}, s_n]} + \cos\left[\frac{\pi(t - s_n)}{2bl^{n + 1}}\right]1_{(s_n, s_{n + 1}]} \qquad(n \geq 1),
\end{gather}
where $s_n = b(1 + l + \ldots + l^n) = b(1 - l^{n+ 1})/(1 - l)$. We note that $\sum_{n = 0}^{\infty}\phi_n^2 = 1$ and that each $\phi_n$ is Lipschitz continuous.
\begin{theorem}[Basic Localization Estimate for Fractional Pauli Operators]
\label{theorem.localization.estimate}
Let $0 < s \leq 1$. For a standard exponential localization of length $b$ and base $l$ one has the following localization estimate,
\begin{gather}
|\mathcal{P}_A|^{2s} \geq \sum_{n = 0}^{\infty}\varphi_n(|\mathcal{P}_A|^{2s} - D_n^{2s})\varphi_n,
\label{equation.localization.estimate}
\end{gather}
where the localization error $D_n^{2s}$ is defined as $\pi^{2s}l^{-2sn}/\left[2^{s}b^{2s}(l^{2s} - 1)\right]$.
\end{theorem}
\begin{proof}
We consider for each $N$ the partition generated by the two elements
\begin{align}
\eta_N \equiv \Biggl(\,\sum_{n \leq N}\varphi_n^2\Biggr)^{1/2}, \qquad \theta_N \equiv \Biggl(\,\sum_{n \geq N + 1}\varphi_n^2\Biggr)^{1/2},
\end{align}
which satisfy $\eta_N^2 + \theta_N^2 = 1$. In order now to prove the assertion we fix a value $N \geq 0$ and perform a single localization on $(p - A)^2$ using $\eta$ and $\theta$, which follows from the standard IMS formula for Schr\"{o}dinger operators (with magnetic fields),
\begin{gather}
(p - A)^2 = \eta_N(p - A)^2\eta_N + \theta_N(p - A)^2\theta_N - |\nabla\eta_N|^2 - |\nabla\theta_N|^2.
\end{gather}
Then we get, by noticing
\begin{gather}
\max(|\nabla\eta_N|, |\nabla\theta_N|) \leq \frac{\pi}{2bl^{N + 1}},
\end{gather}
that
\begin{gather}
(p - A)^2 \geq \eta_N(p - A)^2\eta_N + \theta_N(p - A)^2\theta_N - \frac{\pi^2}{2b^2l^{2(N + 1)}}.
\end{gather}
We will call this error term, $-\pi^2/(2b^2l^{2(N + 1)})$, $-C_{N + 1}^2$. Furthermore,
\begin{align}
\left[\sigma\cdot(p - A)\right]^2 & = (p - A)^2 - \sigma\cdot B\nonumber\\
& \geq \eta_N\left[(p - A)^2 - \sigma\cdot B\right]\eta_N + \theta_N\left[(p - A)^2 - \sigma\cdot B\right]\theta_N - C_{N + 1}^2.
\end{align}
where $B = \nabla\times A$.

From here we obtain the following first localization bound on $|\mathcal{P}_A|^{2s}$, by means of the operator monotonicity of $x \mapsto x^s$ ($0 < s \leq 1$) and the pull-out formula \eqref{equation.pull.out.formula},
\begin{align}
|\mathcal{P}_A|^{2s} & = (\mathcal{P}_A^2)^s \geq (\mathcal{P}_A^2 + C_{N + 1}^2)^s - C_{N + 1}^{2s} \geq \left(\eta_N\left[\sigma\cdot(p - A)\right]^2\eta_N + \theta_N\left[\sigma\cdot(p - A)\right]^2\theta_N\right)^s - C_{N + 1}^{2s}\nonumber\\
& \geq \eta_N|\mathcal{P}_A|^{2s}\eta_N + \theta_N|\mathcal{P}_A|^{2s}\theta_N - C_{N + 1}^{2s}\nonumber\\
& = \eta_N\left(|\mathcal{P}_A|^{2s} - C_{N + 1}^{2s}\right)\eta_N + \theta_N\left(|\mathcal{P}_A|^{2s} - C_{N + 1}^{2s}\right)\theta_N.
\end{align}
We iterate this bound once more in the first parenthesis, thereby obtaining
\begin{align}
& \eta_N\left(|\mathcal{P}_A|^{2s} - C_{N + 1}^{2s}\right)\eta_N\nonumber\\
\geq & \, \eta_N\left[\eta_{N - 1}\left(|\mathcal{P}_{A}|^{2s} - C_{N}^{2s}\right)\eta_{N - 1} + \theta_{N - 1}\left(|\mathcal{P}_A|^{2s} - C_{N}^{2s}\right)\theta_{N - 1} - C_{N + 1}^{2s}\right]\eta_N\nonumber\\
= & \, \eta_N\left[\eta_{N - 1}\left(|\mathcal{P}_A|^{2s} - C_{N}^{2s} - C_{N + 1}^{2s}\right)\eta_{N - 1} + \theta_{N - 1}\left(|\mathcal{P}_A|^{2s} - C_N^{2s} - C_{N + 1}^{2s}\right)\theta_{N - 1}\right]\eta_N.
\end{align}
In this way we find
\begin{align}
|\mathcal{P}_A|^{2s} \geq \, & \eta_{N - 1}\left(|\mathcal{P}_A|^{2s} - C_N^{2s} - C_{N + 1}^{2s}\right)\eta_{N - 1} + \varphi_N\left(|\mathcal{P}_A|^{2s} - C_N^{2s} - C_{N + 1}^{2s}\right)\varphi_N\nonumber\\
& + \theta_N\left(|\mathcal{P}_A|^{2s} - C_{N + 1}^{2s}\right)\theta_N,
\end{align}
where we used the relations $\eta_N\eta_{N - 1} = \eta_{N - 1}$ and $\eta_N\theta_{N - 1} = \varphi_N$. Therefore, by continuing the iteration we eventually find that
\begin{gather}
|\mathcal{P}_A|^{2s} \geq \sum_{n = 0}^N \varphi_n\left(|\mathcal{P}_A|^{2s} - \sum_{m = n}^{N + 1}C_m^{2s}\right)\varphi_n + \theta_N\left(|\mathcal{P}_A|^{2s} - C_{N + 1}^{2s}\right)\theta_N.
\end{gather}
We furthermore notice that
\begin{gather}
\sum_{m = n}^{N + 1}C_m^{2s} \leq \sum_{m = n}^{\infty}C_m^{2s} = \frac{\pi^{2s}}{2^{s}b^{2s}}\sum_{m = n}^{\infty}l^{-2sm} = \frac{\pi^{2s}l^{-2sn}}{2^{s}b^{2s}(l^{2s} - 1)} \equiv D_n^{2s},
\end{gather}
and then conclude that
\begin{gather}
|\mathcal{P}_A|^{2s} \geq \sum_{n = 0}^N \varphi_n\left(|\mathcal{P}_A|^{2s} - D_n^{2s}\right)\varphi_n + \theta_N\left(|\mathcal{P}_{A}|^{2s} - C_{N + 1}^{2s}\right)\theta_N.
\end{gather}
We now note that for a fixed $f$ in $C_0^{\infty}$, $\varphi_n f = \theta_n f = 0$ for sufficiently large $n$, and therefore
\begin{gather}
\left(f, |\mathcal{P}_A|^{2s} f\right) \geq \sum_{n = 0}^{\infty}\left(f, \varphi_n\left(|\mathcal{P}_A|^{2s} - D_n^{2s}\right)\varphi_n f\right),
\end{gather}
which completes the proof of the assertion.
\end{proof}
We compile now two results on the trace of operators of the kind that will be dealt with in the next section.
\begin{lemma}[Estimate on the Trace of a Sum]
\label{lemma.trace.sum}
For a collection of bounded-below self-adjoint operators $T_n$ such that the sum $\sum_{n = 1}^{\infty}T_n$ is bounded-below one has the estimate
\begin{gather}
\text{\textnormal{Tr}}\left(\sum_{n = 1}^{\infty}T_n\right)_{-} \leq \sum_{n = 1}^{\infty}\text{\textnormal{Tr}}\,(T_n)_{-}.
\end{gather}
\end{lemma}
\begin{proof}
This follows immediately from the variational principle for the trace of the negative part of a bounded-below self-adjoint operator $T$, namely
\begin{gather}
\text{Tr}\, T_{-} = -\min_{0 \leq \rho \leq 1}\text{Tr}\,\left(T\rho\right)
\end{gather}
(here the minimum is taken over all bounded self-adjoint operators $\rho$ with the property that $0 \leq \rho \leq 1$), since
\begin{align}
&\text{Tr}\,\left(\sum_{n = 1}^{\infty}T_n\right)_{-} = -\min_{0 \leq \rho \leq 1}\text{Tr}\,\left(\sum_{n = 1}^{\infty}T_n \rho\right) = - \min_{0 \leq \rho \leq 1}\sum_{n = 1}^{\infty}\text{Tr}\,\left(T_n \rho\right)\nonumber\\
\leq & \, -\sum_{n = 1}^{\infty}\min_{0 \leq \rho \leq 1}\text{Tr}\,\left(T_n\rho\right) = \sum_{n = 1}^{\infty}\text{Tr}\,(T_n)_{-}.
\end{align}
\end{proof}
\begin{lemma}[Estimate on the Trace of a Sum of Products of Operators]
\label{lemma.trace.sum.products}
For a bounded-below self-adjoint operator $T$ and bounded, positive self-adjoint operators $S_n$ such that $\sum_{n = 1}^{\infty}S_n \leq 1$,
\begin{gather}
\text{\textnormal{Tr}}\left(\sum_{n = 1}^{\infty}S_nT S_n\right)_{-} \leq \sum_{n = 1}^{\infty}\text{\textnormal{Tr}}\left(S_nT S_n\right)_{-} \leq \text{\textnormal{Tr}}\,T_{-}.
\end{gather}
\end{lemma}
\begin{proof}
The first inequality follows immediately from the previous lemma. For the second inequality, we see easily that
\begin{align}
& \sum_{n = 1}^{\infty}\text{Tr}\,\left(S_n T S_n\right)_{-} = -\sum_{n = 1}^{\infty}\min_{0 \leq \rho \leq 1}\left(S_n T S_n\rho\right) = -\sum_{n = 1}^{\infty}\text{Tr}\,\left(S_n T S_n \rho_n\right)\nonumber\\
= \, & -\sum_{n = 1}^{\infty}\text{Tr}\,\left(T S_n \rho_n S_n\right) = -\text{Tr}\,\left[T\sum_{n = 1}^{\infty}S_n\rho_n S_n\right],
\end{align}
for some operators $\rho_n$. Now, since $\rho_n \leq 1$,
\begin{gather}
\sum_{n = 1}^{\infty} S_n\rho_n S_n \leq \sum_{n = 1}^{\infty}S_n^2 \leq 1,
\end{gather}
and therefore,
\begin{gather}
-\text{Tr}\,\left[T\sum_{n = 1}^{\infty}S_n\rho_n S_n\right] \leq -\min_{0 \leq \rho \leq 1}\text{Tr}\,\left(T\rho\right) = \text{Tr}\,T_{-}.
\end{gather}
\end{proof}
\end{section}

\begin{section}{Pauli-Hardy-Lieb-Thirring Inequalities}
\label{section.proof.main.result}
In this section we shall utilize all the estimates proven in the previous two sections to prove the Hardy-Lieb-Thirring inequalities for the fractional Pauli operator $|\mathcal{P}_A|^{2s}$ mentioned in the introduction, Theorems \ref{theorem.main.result} and \ref{theorem.second.main.result}. We shall begin with the proof of Theorem \ref{theorem.second.main.result}.
\begin{proof}[Proof of Theorem \ref{theorem.second.main.result}]
We first notice that the case $s = 1$ follows immediately from the estimate $\mathcal{P}_A^2 = p_A^2 - \sigma\cdot B \geq p_A^2 - |B|$ and the use of inequality \eqref{inequality.ekholm.frank}. We then focus on the case $0 < s < 1$. We start by noticing that, by the localization estimate \eqref{equation.localization.estimate} and Lemma \ref{lemma.trace.sum},
\begin{align}
\text{Tr}\,\left(|\mathcal{P}_A|^{2s} - \frac{C_s}{|x|^{2s}} + V\right)_{-} & \leq \, \text{Tr}\,\left(\sum_{n = 0}^{\infty}\varphi_n\left(|\mathcal{P}_A|^{2s} - D_n^{2s} - \frac{C_s}{|x|^{2s}} + V\right)\varphi_n\right)_{-}\nonumber\\
& \leq \, \text{Tr}\,\left(\varphi_0\left[|\mathcal{P}_A|^{2s} - \left(D_0^{2s} + \frac{C_s}{|x|^{2s}} - V\right)\right]\varphi_0\right)_{-}\nonumber\\
& \quad + \sum_{n = 1}^{\infty}\text{Tr}\,\left(\varphi_n\left[|\mathcal{P}_A|^{2s} - \left(D_n^{2s} + \frac{C_s}{|x|^{2s}} - V\right)\right]\varphi_n\right)_{-}.\label{phlt.first.splitting}
\end{align}
We may rewrite the series appearing in \eqref{phlt.first.splitting} as
\begin{gather}
\sum_{n = 1}^{\infty}\text{Tr}\,\left(\varphi_n\left[|\mathcal{P}_A|^{2s} - \chi_n\left(D_n^{2s} + \frac{C_s}{|x|^{2s}} - V\right)\right]\varphi_n\right)_{-},\label{phlt.series}
\end{gather}
where $\chi_n$ is the indicator function of the support of $\varphi_n$. Furthermore, we notice that for $n \geq 1$, since $|x| \leq bl^{n + 2}/(l - 1)$ in the support of $\varphi_n$,
\begin{gather}
\chi_n D_n^{2s} = \frac{\pi^{2s}}{2^s b^{2s}(l^{2s} - 1)}\frac{\chi_n}{l^{2sn}} \leq \frac{\pi^{2s}l^{4s}}{2^s(l^{2s} - 1)(l - 1)^{2s}}\frac{\chi_n}{|x|^{2s}} \equiv (E_s - C_s)\frac{\chi_n}{|x|^{2s}},
\end{gather}
which allows us to bound the series as
\begin{align}
& \sum_{n = 1}^{\infty}\text{Tr}\,\left(\varphi_n\left[|\mathcal{P}_A|^{2s} - \chi_n\left(\frac{E_s}{|x|^{2s}} - V\right)\right]\varphi_n\right)_{-}\nonumber\\
 = \, & \sum_{n = 1}^{\infty}\text{Tr}\,\left(\varphi_n\left[|\mathcal{P}_A|^{2s} - \omega\left(\frac{E_s}{|x|^{2s}} - V\right)\right]\varphi_n\right)_{-},
\end{align}
with $\omega$ as the indicator function of the region $|x| \geq b$. Then, by Lemma \ref{lemma.trace.sum.products}, this is further bounded above by
\begin{align}
\text{Tr}\,\left[|\mathcal{P}_A|^{2s} - \omega\left(\frac{E_s}{|x|^{2s}} - V\right)\right]_{-} \leq \text{Tr}\,\left[|\mathcal{P}_A|^{2s} - \omega\left(\frac{E_s}{|x|^{2s}} - V\right)_{+}\right]_{-}.
\label{phlt.second.term}
\end{align}
At this point the BKS inequality, Equation \eqref{inequality.BKS}, will prove crucial, since it will allow us to bound the right side of \eqref{phlt.second.term} from above by something manageable without introducing any extranous terms into the trace in question,
\begin{align}\label{eq:LT-afterBKS}
\text{Tr}\,\left[|\mathcal{P}_A|^2 - \omega\left(\frac{E_s}{|x|^{2s}} - V\right)_{+}^{1/s}\right]_{-}^s \leq \, & \text{Tr}\,\left[p_A^2 - |B| - \omega\left(\frac{E_s}{|x|^{2s}} - V\right)_{+}^{1/s}\right]_{-}^s\nonumber\\
\leq & \, L_s\int_{\mathbb{R}^3}\left[|B| + \omega\left(\frac{E_s}{|x|^{2s}} - V\right)_{+}^{1/s}\right]^{s + 3/2}\,dx,
\end{align}
where in the last line we used the Lieb-Thirring inequality \eqref{equation.lieb.thirring} (with magnetic fields).

We now focus on the first term in \eqref{phlt.first.splitting}, namely
\begin{gather}
\text{Tr}\,\left(\varphi_0\left[|\mathcal{P}_A|^{2s} - \left(D_0^{2s} + \frac{C_s}{|x|^{2s}} - V\right)\right]\varphi_0\right)_{-} = \text{Tr}\,\left(\varphi_0\left[|\mathcal{P}_A|^{2s} - \frac{C_s}{|x|^{2s}} - \chi_0\left(D_0^{2s} - V\right)\right]\varphi_0\right)_{-}.
\label{phlt.trace.near.origin}
\end{gather}
By using Theorem \ref{theorem.quadratic.form} in \eqref{phlt.trace.near.origin} we obtain
\begin{align}
& \text{Tr}\,\left(\varphi_0\left[|\mathcal{P}_A|^{2s} - \frac{C_s}{|x|^{2s}} - \chi_0\left(D_0^{2s} - V\right)\right]\varphi_0\right)_{-}\nonumber\\
\leq & \, \lambda\,\text{Tr}\,\left(\varphi_0\left[|p|^{2s} - \frac{C_s}{|x|^{2s}} - \frac{\chi_0}{\lambda}\left(\mathcal{L}\int_{\mathbb{R}^3}|B|^{s + 3/2}\,dx + D_0^{2s} - V\right)\right]\varphi_0\right)_{-}\nonumber\\
\leq & \, \lambda\,\text{Tr}\,\left(|p|^{2s} - \frac{C_s}{|x|^{2s}} - \frac{\chi_0}{\lambda}\left(\mathcal{L}\int_{\mathbb{R}^3}|B|^{s + 3/2}\,dx + D_0^{2s} - V\right)\right)_{-}\nonumber\\
\leq & \, \lambda^{-3/(2s)}H_s\int_{\text{supp}\,\varphi_0}\left(\mathcal{L}\int_{\mathbb{R}^3}|B|^{s + 3/2}\,dx + D_0^{2s} - V\right)_{+}^{1 + 3/(2s)}\,dx,
\end{align}
where $H_s$ is the constant in the Hardy-Lieb-Thirring inequality for the operator $|p|^{2s}$ (see Equation \eqref{equation.HLT.inequality}). $\mathcal{L}$ depends on $\lambda$ and $s$ (see the statement of Theorem \ref{theorem.quadratic.form}). By assembling all the pieces together we then obtain,
\begin{align}
\text{Tr}\,\left(|\mathcal{P}_A|^{2s} - \frac{C_s}{|x|^{2s}} - V\right)_{-} \leq & \, L_s\int_{\mathbb{R}^3}\left[|B| + \omega\left(\frac{E_s}{|x|^{2s}} - V\right)_{+}^{1/s}\right]^{s + 3/2}\,dx\nonumber\\
& + \, \sigma^{-3/(2s)}H_s\int_{\text{supp}\,\varphi_0}\left(\mathcal{L}\int_{\mathbb{R}^3}|B|^{s + 3/2}\,dx + D_0^{2s} - V\right)_{+}^{1 + 3/(2s)}\,dx.
\label{phlt.inequality.pre.final}
\end{align}
By using now the bounds $(a + b)_+ \leq a_+ + b_+$ and $(a + b)^{x} \leq 2^{x - 1}(a^x + b^x)$ ($x \geq 1$; $a, b \geq 0$) we obtain that the first term in \eqref{phlt.inequality.pre.final} is bounded above by
\begin{align}
C_1\int_{\mathbb{R}^3}|B|^{s + 3/2}\,dx + C_2 b^{-2s} + C_3\int_{\mathbb{R}^3}V_-^{1 + 3/(2s)}\,dx,
\end{align}
for some constants $C_1, C_2$ and $C_3$ (which depend exclusively on $s$ and the base $l$). For now $b$ will be left arbitrary -- it will be optimized at the end. In complete analogy, we find that the second term in \eqref{phlt.inequality.pre.final} is bounded above by
\begin{gather}
C_4 b^3\left(\int_{\mathbb{R}^3}|B|^{s + 3/2}\,dx\right)^{1 + 3/(2s)} + C_5 b^{-2s} + C_6\int_{\mathbb{R}^3}V_-^{1 + 3/(2s)}\,dx,
\end{gather}
where again $C_4, C_5$ and $C_6$ are constants depending only on $s$ and $l$ (in addition to tunable parameters). After grouping terms together and optimizing the length scale $b$, yielding a number proportional to $\left(\int |B|^{s + 3/2}\right)^{-1/(2s)}$, we find the bound
\begin{gather}
C_7\int_{\mathbb{R}^3}|B|^{s + 3/2}\,dx + C_8\int_{\mathbb{R}^3}V_-^{1 + 3/(2s)}\,dx,
\end{gather}
as claimed.
\end{proof}
\begin{proof}[Proof of Theorem \ref{theorem.main.result}]
The case $1/2 \leq s <1$ follows by modifying the proof of Theorem \ref{theorem.second.main.result} in the following manner: on the left side of Equation \eqref{phlt.second.term} we apply Theorem \ref{theorem.pauli.lieb.thirring.power.2} directly, obtaining
\begin{align}
\text{\textnormal{Tr}}\left[|\mathcal{P}_A|^{2s} - \omega\left(\frac{E_s}{|x|^{2s}} - V\right)\right]_{-} \leq & \, F_1\int_{|x| > b}\left(\frac{E_s}{|x|^{2s}} - V\right)_+^{1 + 3/(2s)}\,dx\nonumber\\
& \, + F_2\left(\int |B|^2\,dx\right)^{3/4}\left(\int_{|x| > b}\left(\frac{E_s}{|x|^{2s}} - V\right)_+^4\,dx\right)^{1/4},
\label{inequality.PLT.power.2.first.used.main.result}
\end{align}
where $F_1$ and $F_2$ are constants. We furthermore bound the right side of \eqref{inequality.PLT.power.2.first.used.main.result} as we did with the first term in \eqref{phlt.inequality.pre.final}, obtaining
\begin{gather}
F_3b^{-2s} + F_4\int_{\mathbb{R}^3}V_{-}^{1 + 3/(2s)}\,dx + F_5\left(\int_{\mathbb{R}^3}|B|^2\,dx\right)^{3/4}b^{-2s + 3/4} + F_6\left(\int_{\mathbb{R}^3}|B|^2\,dx\right)^{3/4}\left(\int_{\mathbb{R}^3}V_{-}^4\,dx\right)^{1/4}.
\label{inequality.proof.phlt.power.2.estimate.outer.part}
\end{gather}
We then use Estimate \eqref{inequality.quadratic.forms.second}, which has a power of 2 for the magnetic field, in Equation \eqref{phlt.trace.near.origin}. We obtain in this case
\begin{align}
& \text{\textnormal{Tr}}\left(\varphi_0\left[|\mathcal{P}_A|^{2s} - \frac{C_s}{|x|^{2s}} - \chi_0\left(D_0^{2s} - V\right)\right]\varphi_0\right)_-\nonumber\\
\leq & \, \lambda\text{\textnormal{Tr}}\left(\varphi_0\left[\left(|p|^{2s} - \frac{C_s}{|x|^{2s}}\right) - \frac{\chi_0}{\lambda}\left(\mathcal{M}\left(\int_{\mathbb{R}^3}|B|^2\,dx\right)^{2s} + D_0^{2s} - V\right)\right]\varphi_0\right)_-\nonumber\\
\leq & \, \lambda\text{\textnormal{Tr}}\left[\left(|p|^{2s} - \frac{C_s}{|x|^{2s}}\right) - \frac{\chi_0}{\lambda}\left(\mathcal{M}\left(\int_{\mathbb{R}^3}|B|^2\,dx\right)^{2s} + D_0^{2s} - V\right)\right]_-\nonumber\\
\leq & \, \lambda^{-3/(2s)}H_s\int_{\text{supp}{\varphi_0}}\left[\mathcal{M}\left(\int_{\mathbb{R}^3}|B|^2\,dx\right)^{2s} + D_0^{2s} - V\right]_+^{1 + 3/(2s)}\,dx,
\end{align}
and, bounding this just as we did with the second term on the right side of inequality \eqref{phlt.inequality.pre.final}, we get
\begin{gather}
F_7b^3\left(\int_{\mathbb{R}^3}|B|^2\,dx\right)^{2s + 3} + F_8 b^{-2s} + F_9\int_{\mathbb{R}^3}V_-^{1 + 3/(2s)}\,dx.
\label{inequality.proof.phlt.power.2.estimate.inner.part}
\end{gather}
Upon optimizing the third term in \eqref{inequality.proof.phlt.power.2.estimate.outer.part} plus the first term in \eqref{inequality.proof.phlt.power.2.estimate.inner.part}, yielding a value for $b$ proportional to $\left(\int |B|^2\right)^{-1}$, we get finally
\begin{align}
\text{Tr}\,\left(|\sigma\cdot p_A|^{2s} - \frac{C_s}{|x|^{2s}} + V\right)_- \leq & \, F_{10}\int_{\mathbb{R}^3}V_-^{1 + 3/(2s)}\,dx + F_{11}\left(\int_{\mathbb{R}^3}|B|^2\,dx\right)^{2s}\nonumber\\
& \, + F_{12}\left(\int_{\mathbb{R}^3}|B|^2\,dx\right)^{3/4}\left(\int_{\mathbb{R}^3}V_-^4\,dx\right)^{1/4}.
\end{align}
as claimed in the statement of Theorem \ref{theorem.main.result} for $1/2 \leq s < 1$. In order to prove Theorem \ref{theorem.main.result} with the power $s = 1$, we use Estimate \eqref{inequality.quadratic.form.s.1} with $r = 2$, in lieu of Estimate \eqref{inequality.quadratic.forms.second}, and then use the Hardy-Lieb-Thirring inequality \eqref{equation.HLT.inequality} with magnetic fields. The proof is otherwise nearly identical to the one for $1/2 \leq s < 1$.
\end{proof}
\end{section}

\begin{appendix}
\begin{section}{Appendix}
\begin{subsection}{The Functions Appearing in the Theorems}
In this first part of the appendix we shall provide the functions that were not made explicit in Theorems~\ref{theorem.estimate.I} and \ref{theorem.estimate.II}. We refrained from writing them out explicitly at the outset, as their inclusion would make the statements of the theorems rather cumbersome and obscure the main content of theirs, which is that the form difference of certain operators can be estimated in a very precise and useful way.
\begin{subsubsection}{Functions in the First Estimate}
In what follows we will encounter a total of 7 variables, $s, u, r, \beta, \delta, \varepsilon, \gamma$, whose domain is $0 < s < 1, 0 \leq u \leq 1, r > 3/2, 3(1 - u) < 2r(1 - s), \beta > 0, 0 < \delta < 1, \varepsilon > 0, 0 < \gamma < 1$. In the first estimate there were two main functions, namely $\Theta$ and $\Omega$. The first one is defined as
\begin{gather}
\Theta(s, r) \equiv \inf_{\beta, \delta}\left[\frac{C_s}{s}\beta^s + I(s, 0, r, \delta)\beta^{s + 3/2r - 1} + J(s, 0, r, \delta)\beta^{s + 3/4r - 1}\right],
\end{gather}
where
\begin{align}
C_{s} \equiv & \, \left(\int_0^{\infty}\frac{\,da}{a^{1 - s}(1 + a)}\right)^{-1} = \frac{\sin(\pi s)}{\pi},\nonumber\\
I(s, u, r, \delta) \equiv & \, C_s D_u^{3/r}S^{3/r}2^{-1}\left[2r(1 - s) - 3(1 - u)\right]r^{-1}(1 - \delta)^{-3/4r},\nonumber\\
J(s, u, r, \delta) \equiv & \, C_s D_u^{3/2r}S^{3/2r}2^{-2}\omega(\delta, r)^{3/4r}\left[4r(1 - s) - 3(1 - u)\right]r^{-1},\nonumber\\
D_u \equiv & \, \max_{x \in [0, \infty)}\frac{x^{1 - u}}{x^2 + 1} = 2^{-1}(1 - u)^{(1 - u)/2}(1 + u)^{(1 + u)/2},
\end{align}
and $S$ and $\omega$ are as in the Pauli-Sobolev Inequality, Lemma \ref{lemma.pauli.sobolev}, whose definitions we repeat here, for completeness,
\begin{align}
S & \equiv \frac{1}{\sqrt{3}}\left(\frac{2}{\pi}\right)^{2/3},\\
\omega(\delta, r) & \equiv \frac{S^{4r/(2r - 3)}3^{3/(2r - 3)}(2r - 3)}{(1 - \varepsilon)(2r)^{2r/(2r - 3)}\varepsilon^{3/(2r - 3)}}.
\end{align}
In addition, $\Omega$ is explicitly given by
\begin{gather}
\Omega(s, u, r, \varepsilon) \equiv \inf_{\beta, \delta, \gamma}\left[\frac{2^{s/u - 1}\Upsilon^{s/u}(1 - \gamma)^{s/u}\varepsilon}{2^{s/u - 1}\left(1 + \gamma\right)^{s/u} + (1 - \gamma)^{s/u}} + \Lambda\right],
\end{gather}
which is defined through the following additional functions
\begin{align}
\Lambda(s, u, r, \beta, \delta, \varepsilon, \gamma) \equiv & \, \frac{C_s}{s}\beta^s + \frac{2rs - 3u}{2rs}\left(\frac{3u}{8rs\eta}\right)^{3u/2rs}I(s, u, r, \delta)^{2rs/(2rs - 3u)}\beta^{\left[-2rs(1 - s) + 3s(1 - u)\right]/(2rs - 3u)}\nonumber\\
& \, + \frac{4rs - 3u}{4rs}\left(\frac{3u}{2rs\eta}\right)^{3u/(4rs - 3u)}J(s, u, r, \delta)^{4rs/(4rs - 3u)}\beta^{\left[3s(1 - u) - 4rs(1 - s)\right]/(4rs - 3u)},
\end{align}
with
\begin{gather}
\eta(s, u, \varepsilon, \gamma) \equiv \frac{\varepsilon(1 - \gamma)^{s/u}}{2^{s/u - 1}(1 + \gamma)^{s/u} + (1 - \gamma)^{s/u}},
\end{gather}
and
\begin{align}
\Upsilon(s, r, \beta, \delta, \gamma) \equiv & \, \frac{C_s}{s}\beta^s + \frac{2r - 3}{2r}\left(\frac{3}{8r\gamma}\right)^{3/2r}I(s, s, r, \delta)^{2r/(2r - 3)}\beta^{-(1 - s)}\nonumber\\
& \, + \frac{4r - 3}{4r}\left(\frac{3}{2r\gamma}\right)^{3/(4r - 3)}J(s, s, r, \delta)^{4r/(4r - 3)}\beta^{-(1 - s)}.
\end{align}
\end{subsubsection}
\begin{subsubsection}{Functions in the Second Estimate}
In the functions in the second estimate a total of 5 variables appear, $s, u, r, \varepsilon, \delta$, with domain $0 < s < 1, 0 \leq u \leq 1, 3/2 < r < 3, 3(1 - u) < 2r(1 - s), \varepsilon > 0, 0 < \delta < 1$. The main functions are as follows:
\begin{align}
\mathcal{I}(s, r) & \equiv C_s(2r - 3)S^{2s(3 - r)/(2r - 3)}\left[r(1 - s)s\right]^{-1}2^{-(6s + 2r - 3)/(2r - 3)}N_r^{2rs/(2r - 3)},\\
\mathcal{J}(s, u, r, \varepsilon) & \equiv \inf_{\delta}\left\{\mathcal{A}_{s, u, r}\varepsilon^{-6u/(2r - 3)}\left[1 + 2^{s/u - 1}\delta^{s/u}(1 - \delta)^{-s/u}\right]^{6u/(2r - 3)}\right.\nonumber\\
& \qquad\qquad\left. + \mathcal{A}_{u, u, r}^{s/u} 2^{s/u - 1}\delta^{-6s/(2r - 3)}\varepsilon\left[(1 - \delta)^{s/u} + 2^{s/u - 1}\delta^{s/u}\right]^{-1}\right\},
\end{align}
where $N_r$ is the function appearing in the generalized Sobolev inequality $\|A\|_{3r/(3 - r)} \leq N_r\|\nabla\times A\|_r$ (see Theorem \ref{theorem.sobolev.inequality.rotors}). In $\mathcal{J}$ an additional function $\mathcal{A}$ appears; it is given explicitly by
\begin{align}
\mathcal{A}_{s, u, r} & \equiv C_s^{[2r - 3(1 - u)]/(2r - 3)}\left[2r(1 - s) - 3(1 - u)\right]^{-[2r - 3(1 - u)]/(2r - 3)}\left[2r - 3(1 - u)\right]^{-3u/(2r - 3)}\nonumber\\
& \quad\, \times s^{-[2r - 3(1 - u)]/(2r - 3)}S^{2s(3 - r)/(2r - 3)}D_u^{6s/(2r - 3)}N_r^{2rs/(2r - 3)}(2r - 3)\nonumber\\
& \quad\, \times\left[r(3 - r)u^2\right]^{3u/(2r - 3)}2^{(6u + 2r - 3)/(2r - 3)}.
\end{align}
\end{subsubsection}
\end{subsection}
\begin{subsection}{General Sobolev Inequality for Rotors}
In this subsection of the appendix we will provide the Sobolev inequality for rotors that was used in the proof of the second estimate in Section \ref{section.main.estimates}, Theorem \ref{theorem.estimate.II}.
\begin{theorem}[Sobolev Inequality for Rotors]
\label{theorem.sobolev.inequality.rotors}
For any $r > 3/2$ and $A \in C_0^{\infty}$ with the property that $\nabla\cdot A = 0$,
\begin{gather}
\|A\|_r \leq N_r\|\nabla\times A\|_{3r/(3 + r)},
\end{gather}
for some constant $N_r$.
\end{theorem}
\begin{proof}
Let $A$ be in $C_0^{\infty}$ with the property that $\nabla\cdot A = 0$. Let us first establish the formula
\begin{gather}
A(x) = \frac{1}{4\pi}\int_{\mathbb{R}^3}\frac{x - y}{|x - y|^3}\times \left(\nabla \times A(y)\right)\,dy.
\label{equation.definition.A}
\end{gather}
By means of the formula $\nabla\times(f a) = \nabla f \times a$ ($f$ a scalar function of $x$, $a$ a constant vector), we find that the integral on the right side of \eqref{equation.definition.A} is equal to
\begin{align}
-\nabla\times\int_{\mathbb{R}^3}\nabla\times A(y)|x - y|^{-1}\,dy = & -\nabla\times\int_{\mathbb{R}^3}\nabla\times A(x - y)|y|^{-1}\,dy\nonumber\\
= & -\nabla\times\left(\nabla\times\int_{\mathbb{R}^3}A(x - y)|y|^{-1}\,dy\right).
\end{align}
And then, since $\nabla\times\left(\nabla\times F\right) = \nabla\left(\nabla\cdot F\right) - \Delta F$,
\begin{align}
&\frac{1}{4\pi}\int_{\mathbb{R}^3}\frac{x - y}{|x - y|^3}\times \left(\nabla \times A(y)\right)\,dy\nonumber\\
= \, & \Delta\left((4\pi)^{-1}\int_{\mathbb{R}^3}A(x - y)|y|^{-1}\,dy\right) - \frac{1}{4\pi}\int_{\mathbb{R}^3}\nabla\left(\nabla\cdot A(x - y)\right)|y|^{-1}\,dy\nonumber\\
= \, & \Delta\left(\Delta^{-1}A\right)(x) = A(x),
\end{align}
which proves the equality. Now, note that we have the pointwise estimate
\begin{gather}
|A(x)| \leq \frac{1}{4\pi}\int_{\mathbb{R}^3}\frac{\left|\nabla\times A(y)\right|}{|x - y|^2}\,dy
\end{gather}
and that $|y|^{-2}$ is in $L^{3/2}_w$. We then have, by the weak Young's inequality, for any $1 < p < 3$ and $r > 3/2$ such that $1/p = 1/r + 1/3$,
\begin{gather}
\left\Vert A\right\Vert_r \leq (4\pi)^{-1}\left\Vert\nabla\times A(y) *|y|^{-2}\right\Vert_r \leq C_p\|\nabla\times A\|_p,
\end{gather}
which is a general Sobolev inequality for rotors.
\end{proof}
\end{subsection}
\end{section}
\end{appendix}

\end{document}